\documentclass[11pt,letterpaper]{report}
\usepackage{times,mathptmx}
\DeclareMathAlphabet{\mathcal}{OMS}{cmsy}{m}{n}
\usepackage{fullpage}
\usepackage{amsmath,amsfonts,amssymb,amsthm,amscd}
\usepackage{tabularx,booktabs}
\usepackage{graphicx,epstopdf}
\usepackage{hyperref,cite,url}
\hypersetup{pdfpagemode=UseNone,pdfstartview=}
\usepackage{mathdots}
\newcommand{\vs}{\vspace{1.5mm}}
\theoremstyle{plain} % plain
\newtheorem{theorem}{Theorem}[section]
\newtheorem{lemma}[theorem]{Lemma}

\theoremstyle{definition} % definition
\newtheorem{definition}{Definition}[section]

\newtheorem{example}{Example}
\theoremstyle{remark} % remark

\newcommand{\Z}{\mathbb{Z}}
\newcommand{\bits}{\{0,1\}}

\newcommand{\db}{\displaybreak[0]}
\usepackage{algorithm}
\usepackage[noend]{algpseudocode}
\makeatletter
\def\BState{\State\hskip-\ALG@thistlm}
\makeatother

\begin{document}

%\title{Permutation Generators Based on Unbalanced Feistel Network:\\
%Analysis of the Conditions of Pseudorandomness}
%
%\author{
%    A Thesis for the Degree of Master of Engineering\\ \\
%    Kwangsu Lee\\ \\
%    Department of Computer Science,\\
%    Korea Advanced Institute of Science and Technology
%}
%
%\date{February 2000}

%\maketitle

\begin{titlepage}
\begin{center}
    \vspace*{2.0cm}
    \LARGE{Permutation Generators Based on Unbalanced Feistel Network:
    Analysis of the Conditions of Pseudorandomness}\footnote{
    Advisor: Kwang-Hyung Lee.
    This is an English translation of the MS Thesis written in Korean.}\\

    \vspace{0.8cm}
    \Large{Kwangsu Lee}

    \vspace{4.0cm}
    \large{A Thesis for the Degree of Master of Science}

    \vfill

    \large{
    Division of Computer Science,\\
    Department of Electrical Engineering and Computer Science,\\
    Korea Advanced Institute of Science and Technology}\\

    \vspace*{0.5cm}
    \large{February 2000}
    \vspace*{2.0cm}
\end{center}
\end{titlepage}

%\newpage

\chapter*{Abstract}

A block cipher is a bijective function that transforms a plaintext to a
ciphertext. A block cipher is a principle component in a cryptosystem because
the security of a cryptosystem depends on the security of a block cipher. A
Feistel network is the most widely used method to construct a block cipher.
This structure has a property such that it can transform a function to a
bijective function. But the previous Feistel network is unsuitable to
construct block ciphers that have large input-output size. One way to
construct block ciphers with large input-output size is to use an unbalanced
Feistel network that is the generalization of a previous Feistel network.
There have been little research on unbalanced Feistel networks and previous
work was about some particular structures of unbalanced Feistel networks. So
previous work didn't provide a theoretical base to construct block ciphers
that are secure and efficient using unbalanced Feistel networks.

In this thesis, we analyze the minimal number of rounds of pseudo-random
permutation generators that use unbalanced Feistel networks. That is, after
categorizing unbalanced Feistel networks as source-heavy structures and
target-heavy structures, we analyze the minimal number of rounds of
pseudo-random permutation generators that use each structure. Therefore, in
order to construct a block cipher that is secure and efficient using
unbalanced Feistel networks, we should follow the results of this thesis.
Additionally, we propose a new unbalanced Feistel network that has some
advantages such that it can extend a previous block cipher with small
input-output size to a new block cipher with large input-output size. We also
analyze the minimum number of rounds of a pseudo-random permutation generator
that uses this structure.

%\vs \noindent {\bf Keywords:} Block cipher, Feistel network, Pseudo-random
%function, Pseudo-random permutation.

%\newpage
\tableofcontents
%\newpage

\chapter{Introduction}

A block cipher is a symmetric key cryptosystem that encrypts a plaintext into
a ciphertext block by block \cite{MenezesOV97}. The block cipher should have
the one-to-one correspondence between plaintexts and ciphertexts. The block
cipher is the most important element in most cryptographic systems. In
particular, it is the essential element used in the implementation of other
cryptographic primitives such as pseudo-random number generators, stream
ciphers, and hash functions \cite{MenezesOV97,Schneier94}. Since the security
of most cryptographic systems depends on the security of block ciphers, a
secure block cipher must be implemented to build a secure cryptosystem. A
secure block cipher is a block cipher where the output value of the block
cipher becomes a random value. That is, when the output value of the block
cipher is random, the block cipher becomes a secure one since it is very hard
for an attacker to guess the plaintext or the secret key of the cipher from
the ciphertext. Mathematically, a block cipher is a secure block cipher if it
is a pseudo-random permutation (PRP) generator \cite{LubyR88}.

The biggest problem of implementing a block cipher is that it is difficult to
implement a function that has the one-to-one correspondence property and the
output randomness property at the same time. The way to solve this problem is
to use a Feistel network structure. A Feistel network structure is a method
that converts an arbitrary function into a one-to-one correspondence
function. This structure was designed by H. Feistel in designing the Lucifer
cipher \cite{Feistel73,FeistelNS75}. That is, if a block cipher is
implemented using a Feistel network, a secure block cipher can be easily
implemented by simply implementing an arbitrary function whose output value
is random. This is because if the output value of an arbitrary function is
random, the Feistel network structure automatically converts it to a
one-to-one correspondence function.

Therefore, most block ciphers are constructed by using a Feistel network
structure or a slightly modified structure of the Feistel network structure
\cite{DES77,ShimizuM88,Schneier94,Rivest95}. In addition, some research has
been done on the security of block ciphers using the Feistel network
structure \cite{BihamS91,MenezesOV97,Matsui94,SadeghiyanP92}. In particular,
there have been numerous studies on pseudo-random permutation generators
after the work of Luby and Rackoff \cite{LubyR88,Schnorr88,ZhengMI90i,
Pieprzyk91,SadeghiyanP91,SadeghiyanP92,Patarin92n,Patarin92h,AiolloV96,
Coppersmith96,Patarin98}. However, the problem of the previous (balanced)
Feistel network is that it is difficult to construct a block cipher with a
large input/output size. That is, when a block cipher that can process a
large size of data at one time is implemented, the input size of a round
function used in the Feistel network structure also increases as the
input/output size of the block cipher increases. In practice, however, the
cost of implementing a round function is proportional to the input size of
the round function. Therefore, the previous (balanced) Feistel network
structure is inappropriate when constructing a block cipher with large
input/output size.

One of the ways to solve this problem is to use an unbalanced Feistel network
structure which is a modification of the previous balanced Feistel network.
An unbalanced Feistel network is a Feistel network structure in which the
size of a target-block combined with the output of a round function and the
size of a source-block which is the input of the round function are different
\cite{AndersonB96,Lucks96,SchneierK96}. Therefore, the round function in an
unbalanced Feistel network structure can be implemented at a lower cost than
a balanced Feistel network since it is possible to control the input size of
the round function in the unbalanced Feistel network structure. Especially,
as the information processing capability of the computer increases, a block
cipher with a large input-output size that is capable of processing a large
amount of information will be needed. Therefore, much research is needed on
unbalanced Feistel networks suitable for implementing block ciphers with
large input-output size. However, there are not many studies on unbalanced
Feistel networks.

There are some studies to implement a pseudo-random permutation generator,
which is a secure block cipher using an unbalanced Feistel network
\cite{Jutla98,NaorR99}. However, these studies failed to show the minimum
number of rounds for a block cipher using an unbalanced Feistel network to
become a secure and efficient block cipher. In fact, the minimum number of
rounds is important when constructing a block cipher because the number of
rounds greatly affects the speed and cost of the block cipher. That is, when
constructing a block cipher, a small number of rounds must be used to
implement a fast block cipher at low cost. It is therefore very important to
determine the minimum number of rounds to be a secure block cipher.

Therefore, in this thesis, we first find the minimum number of rounds for a
pseudo-random permutation generator, which is a secure block cipher using a
Feistel network. Next, we propose a scalable new unbalanced Feistel network
structure and find the minimum number of rounds for this structure to become
a secure block cipher. The advantage of a newly proposed structure is that it
can easily construct a block cipher with a large input-output size using a
previously designed secure block cipher with fixed input-output size. In
other words, there is a lot of analysis and research on a new block cipher in
order to newly design a block cipher with a large input-output size. However,
by using this newly proposed structure, a new block cipher with a large
input-output size can be implemented using the previously analyzed secure
block cipher. So we do not have to do another analysis and research.

The structure of this thesis is as follows. In Chapter 2, we first define a
block cipher, a Feistel network, and a pseudo-random number generator. Then
we summarize the existing studies on pseudo-random permutation generators. In
Chapter 3, we investigate the condition of the number of rounds for a block
cipher using an unbalanced Feistel network to be a pseudo-random permutation
generator. In Chapter 4, we propose a scalable unbalanced Feistel network
structure, and analyze the conditions for a secure block cipher using this
new structure to be a pseudo-random permutation generator. In Chapter 5, we
compare a balanced Feistel network structure, an unbalanced Feistel network
structure, and the newly proposed structure. Finally, in Chapter 6, we
conclude the thesis and present the direction of future research.

\chapter{Preliminaries}

In this chapter, we define the terms used in the thesis and summarize the
studies related to a pseudo-random permutation generator. A block cipher is a
secret-key cryptosystem that processes messages using the same key when
encrypting and decrypting messages. A block cipher is the basis of
cryptographic systems for message authentication, data integrity
verification, and digital signature. Mathematically, a block cipher is a
one-to-one function (permutation) since it must be able to encrypt and
decrypt messages using a secret key. A Feistel network structure is most
commonly used to build block ciphers because it has the advantage of
converting arbitrary functions to permutations. In order for a block cipher
to be a secure block cipher, the output value of the block cipher must be a
random value. That is, when a block cipher becomes a pseudo-random
permutation generator, it becomes a secure block cipher. In this case, the
permutation generator $P$ is pseudo-random, meaning that any efficient
algorithm can not distinguish between an ideal permutation generator and $P$.

This chapter is organized as follows. In Section 2.1, we first define symbols
used in this paper. In Section 2.2, we define a block cipher, which is a
secret-key cryptosystem, and investigate the characteristics of the block
cipher and attack methods for block ciphers. In Section 2.3, we define the
most commonly used Feistel networks for building block ciphers and discuss
the advantages and disadvantages of them. In Section 2.4, we define the
pseudorandomness. In Section 2.5, we finally summarize existing studies on
pseudo-random permutation generators

\section{Notation}

The symbols used in this paper are defined as follows.
\begin{itemize}
\item $I_n$ represents a set of all $n$-bit strings. That is, $\bits^n$.

\item $F: I_s \rightarrow I_t$ is a set of all functions whose inputs are
    $s$-bits and whose outputs are $t$-bits.

\item $F_n$ is a set of functions whose input and output are both $n$-bits
    in size. That is, $F_n$ is $F: I_n \rightarrow I_n$.

\item $P_n$ is a set of permutations (one-to-one functions) whose input and
    output sizes are both $n$-bits. That is, $P_n \subset F_n$

\item $|x|$ is the length of a bit string $x$. That is, $|x|$ is $n$ if the
    bit size of $x$ is $n$-bits.

\item $x \oplus y$ is an exclusive OR (XOR) per bit unit when the bit size
    of $x$ and $y$ is equal.

\item $x \| y$ is a concatenation of two bit strings $x$ and $y$. In this
    case, we have $| x \| y | = |x| + |y|$.

\item $f \circ g$ is the composition of two functions $f$ and $g$ when $f$
    and $g$ are elements of the set $F_n$. That is, $f \circ g(x) =
    f(g(x))$.
\end{itemize}

\section{Block Cipher}

\begin{figure}[t]
\centering
\includegraphics[scale=0.75]{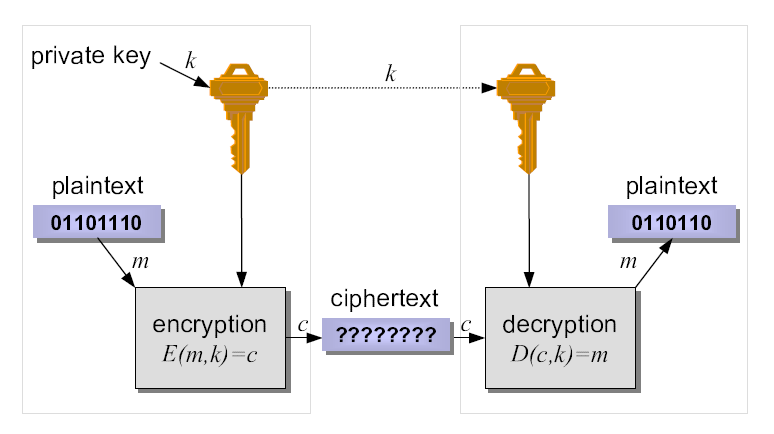}
\caption{The overview of a secret-key cryptosystem}
\label{fig:block_cipher}
\end{figure}

A block cipher can usually be a private-key cryptosystem or a public-key
cryptosystem. However, in this paper, the secret-key cryptosystem is called a
block cipher. The structure of a block cipher is given in Figure
\ref{fig:block_cipher}.

A block cipher is a function that sends an $n$-bit plaintext to an $n$-bit
ciphertext \cite{MenezesOV97,Shannon49}. The function of the block cipher is
specified by an $\ell$-bit secret key. If a plaintext is encrypted and
decrypted again, the original plaintext must be obtained. Therefore, the
block cipher must be a one-to-one function (bijection) for $n$-bit plaintexts
and $n$-bit ciphertexts when a secret key is specified. That is, each secret
key defines a different permutation. The definition of the block cipher is as
follows.

\begin{definition}[Block Cipher] Let $K$ be the set of $\ell$-bit secret keys. The
$n$-bit block cipher is defined as a function $E : I_n \times K \rightarrow
I_n$. We have that $E^{-1}(E(p,k), k) = p$ holds for an arbitrary plaintext
$p$ and a random key $k$ ($k \in K$).
\end{definition}

A true random block cipher is a block cipher that generates all permutations
between the domain and the range \cite{MenezesOV97}. In the $n$-bit block
cipher, the domain corresponds to the set of plaintexts and the number of
plaintexts is $2^n$, so the size of the domain is $2^n$. The range
corresponds to the set of ciphertexts and the number of ciphertexts is $2^n$,
so the size of the range is also $2^n$. Therefore, the number of all
permutations is $2^n!$. In the $n$-bit block cipher, one key specifies one
permutation, so we need $2^n !$ number of secret keys to enumerate all
permutations. That is, the size of a secret key of the ideal block cipher
must be $\log (2^n !) \approx (n - 1.44) 2^n$ bits.

\begin{definition}[Ideal Block Cipher] An ideal block cipher is a block
cipher that implements all $2^n !$ number of one-to-one functions that exist
between $2^n$ number of elements. In this case, each secret key specifies
each one-to-one function.
\end{definition}

It is impossible to actually build an ideal block cipher because it requires
$(n-1.44) 2^n$ bits for a secret key. Thus, in order for a block cipher using
an $\ell$-bit secret key to be secure, the permutations that are specified by
$\ell$-bit secret keys must appear randomly chosen from all $2^n !$
permutations.

The security of a block cipher is measured by the security against the
various attack methods of attackers. The attack on the block cipher is
divided into four categories according to the information that the attacker
can access:
\begin{enumerate}
\item Ciphertext-only attack: The attacker uses only ciphertexts to get the
    secret key of the block cipher.

\item Known-plaintext attack: The attacker uses known plaintexts and
    ciphertexts pairs to find the secret key of the block cipher.

\item Chosen-plaintext attack: The attacker finds the secret key of the
    block cipher by using the pairs of plaintexts and corresponding
    ciphertexts chosen by the attacker.

\item Chosen-ciphertext attack: The attacker finds the secret key using the
    chosen ciphertexts and its corresponding plaintexts.
\end{enumerate}

Differential cryptanalysis and linear cryptanalysis are the most powerful
methods of attacking block ciphers. Differential cryptanalysis is a
chosen-plaintext attack developed by Biham and Shamir \cite{BihamS91}. This
attack method exploits the fact that the probability distribution of the
difference between the input/output pair of a nonlinear function is not
uniform. Linear cryptanalysis is a known-plaintext attack developed by Matsui
\cite{Matsui94}. This attack method extracts the information of a related key
by using a linear approximation of a nonlinear function.

\section{Feistel Network}

A Feistel network is the most commonly used structure for designing a block
cipher. This structure was first used when designing a Lucifer cipher by H.
Feistel \cite{Feistel73,FeistelNS75}. After that, this was used to design
block ciphers such as DES, FEAL, Blowfish, and RC5 \cite{DES77,ShimizuM88,
Schneier94,Rivest95}.

A Feistel network is a method that converts an arbitrary function to a
permutation that is a one-to-one correspondence function. The definition of a
Feistel network is as follows.

\begin{definition}[Feistel Network] \label{def:feistel_network}
For any function $f$ belonging to $F:I_s \rightarrow I_t$, one round Feistel
network is defined as a function $D_f (L \| R) = (R \| L \oplus f(R))$.
Similarly, for the functions $f_1, f_2, \ldots, f_r$, which belong to the set
$F:I_s \rightarrow I_t$, an $r$ rounds Feistel network is defined as a
function $D_{f_r} \circ \cdots \circ D_{f_2} \circ D_{f_1} (L \| R)
\stackrel{def}{=} D_{f_r} \circ \cdots \circ D_{f_2} (D_{f_1} (L \| R))$,
where $|L| = t$, $|R| = s$, $|L| + |R| = n$, and $D_f \in P_n$.
\end{definition}

The above Feistel network can be seen as a permutation. To show that a
function is a one-to-one correspondence function (bijection), we should show
that it is a one-to-one function and an onto function. However, we only need
to show that it is a one-to-one function since the input and output bits of
the Feistel network are the same.

\begin{theorem} \label{thm:1_round_feistel_network}
The function $D_f (L \| R) = (R \| L \oplus f(R))$ is a one-to-one
function.
\end{theorem}

\begin{proof}
If the function $D_f$ is a one-to-one function, then $D_f (x) \neq D_f (y)$
for $x$ and $y$ such that $x \neq y$. If $x = (L_1 \| R_1)$ and $y = (L_2 \|
R_2)$, then $D_f (x) = (R_1 \| L_1 \oplus f(R_1))$ and $D_f(y) = (R_2 \| L_2
\oplus f(R_2))$. Because of $x \neq y$, we consider two cases.
\begin{itemize}
\item Case $R_1 \neq R_2$: $D_f (x) \neq D_f(y)$ by the definition of
    $D_f$.
\item Case $L_1 \neq L_2$ and $R_1 = R_2$: $L_1 \oplus f(R_1) \neq L_2
    \oplus f(R_2)$ since $f(R_1) = f(R_2)$.
\end{itemize}
So the function is a one-to-one function.
\end{proof}

\begin{figure}[t]
\centering
\includegraphics[scale=0.75]{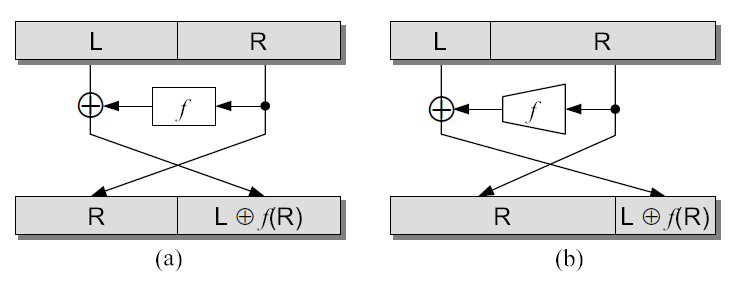}
\caption{Feistel networks: (a) Balanced Feistel network,
(b) Unbalanced Feistel network}
\label{fig:feistel_network}
\end{figure}

A Feistel network is divided into a balanced Feistel network and an
unbalanced Feistel network \cite{SchneierK96}. A balanced Feistel network is
a Feistel network in Definition \ref{def:feistel_network} with the same $L$
and $R$ sizes. In contrast, an unbalanced Feistel network is a Feistel
network in Definition \ref{def:feistel_network} with different $L$ and $R$
sizes ($|L| \neq |R|$). The balanced Feistel network and unbalanced Feistel
network structures are given in Figure \ref{fig:feistel_network}.

The DES algorithm that uses a Feistel network was invented in 1970s and it
has been used as the standard block cipher for 20 years \cite{DES77}. Many
researchers have studied the security of the DES algorithm \cite{BihamS91,
Matsui94,Schneier96}. In addition, many other block ciphers that were
invented after DES were also affected by the DES cipher.

\begin{example}[DES]
The DES algorithm is a 64-bit block cipher with a 56-bit secret key. This
cipher has a 16 rounds balanced Feistel network structure. The $i$th round is
defined as
    $$D_{K_i}(L^{i-1} \| R^{i-1})
    \stackrel{def}{=} (R^{i-1} \| L^{i-1} \oplus P(S(E(R) \oplus K_i)))$$
where $|L^{i-1}| = |R^{i-1}| = 32$, $|K_i| = 48$, $E:I_{32} \rightarrow
I_{48}$ is an expansion function, $S:I_{48} \rightarrow I_{32}$ is a
substitution function, and $P:I_{32} \rightarrow I_{32}$ is a permutation
function.
\end{example}

An unbalanced Feistel network is divided into a source-heavy unbalanced
Feistel network and a target-heavy unbalanced Feistel network
\cite{SchneierK96}. The source-heavy unbalanced Feistel network is a Feistel
network where the size of the block $R$ which is the input to the $F$
function, is greater than the size of the block $L$ which is combined with
the output of the $F$ function ($|R| > |L|$). Whereas the target-heavy
Feistel network is a Feistel network where the size of the block $R$ is
smaller than the size of the block $L$ ($|R| < |L|$). The structures of
source-heavy and target-heavy Feistel networks are given in Figure
\ref{fig:unbalanced_feistel_network}.

\begin{figure}
\centering
\includegraphics[scale=0.75]{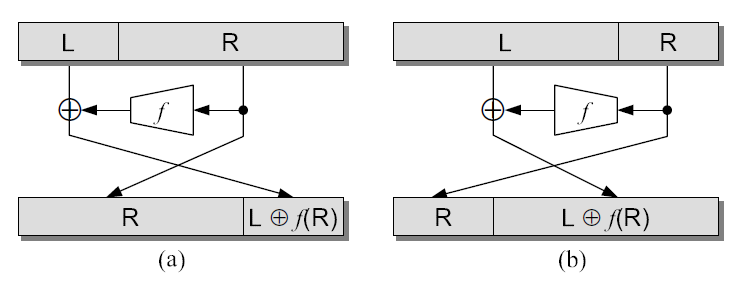}
\caption{Unbalanced Feistel networks: (a) Source-heavy Feistel network,
(b) Target-heavy Feistel network}
\label{fig:unbalanced_feistel_network}
\end{figure}

The MARS algorithm is a block cipher proposed for the advanced encryption
standard (AES) to replace the DES block cipher \cite{BurwickCD+98}. This
cipher uses a target-heavy unbalanced Feistel network.

\begin{example}[MARS]
The MARS algorithm is a 128-bit block cipher proposed for AES and it can have
different size of secret keys. This cipher has a 32 rounds unbalanced Feistel
network. The $i$th round is defined as
    $$D_{K_i}(L_1^{i-1} \| L_2^{i-1} \| L_3^{i-1} \| R^{i-1})
    \stackrel{def}{=} (R^{i-1} \| L_1^{i-1} \boxplus f_1^i(R^{i-1}) \|
    L_2^{i-1} \boxplus f_2^i(R^{i-1}) \| L_3^{i-1} \boxplus f_3^i(R^{i-1}))$$
where $\boxplus$ is the addition operator in mod $2^{32}$.
\end{example}

The main advantage of a block cipher using an unbalanced Feistel network is
that it can select the input size of the $F$ function used in the Feistel
network. In particular, as the amount of data that the computer has to
process due to the development of the Internet, the amount of messages to be
encrypted is also increasing. As the information throughput increases, an
encryption scheme that can process large amounts of data becomes necessary.
Therefore, a block cipher with a large input size is needed. The advanced
encryption standard (AES) selected by National Institute of Standards and
Technology (NIST) also requires a 128-bit block cipher to reflect this demand
\cite{NIST97}.

When implementing a block cipher with a large input value, it is difficult to
implement the $F$ function of a Feistel network if a balanced Feistel network
is used. This is because the cost of implementing the $F$ function is
generally proportional to the size of the input value of the $F$ function.
However, since the input size of the $F$ function can be selected in an
unbalanced Feistel network, it is very effective to use an unbalanced Feistel
network in implementing a large block cipher. In other words, a block cipher
with a large input size using a target-heavy unbalanced Feistel network can
be implemented at a lower cost than a cipher using other Feistel network
structures.

\section{Pseudo-Randomness}

Before defining pseudo-randomness, we first review what two distributions are
computationally equivalent. The computational equivalence of two
distributions means that no effective algorithm can determine that two
distributions are different. In other words, computational
indistinguishability is a criterion for judging the equivalence of two
distributions. Therefore, the pseudo-randomness distribution is a
distribution that can not be distinguished from the truly random distribution
by computation \cite{Goldreich95}.

At this time, it is necessary to define an algorithm for determining whether
two distributions are identical. The definition of an algorithm for
determining identical distributions is defined by the following oracle
machine.

\begin{definition}[Oracle Machine $M$]
An oracle machine $M$ is a Turing machine that has an oracle tape as an
additional tape and has two special states called ``oracle invocation'' and
``oracle appeared''. The oracle machine has an input value of $1^n$ and an
output value of 1. The oracle machine that can access a function $f$ whose
input value is $n$ bits in size is called $M^{f} (1^n)$ and operates as
follows: If the state of the oracle machine is not ``oracle appeared'', then
it operates as the same as a normal Turing machine. If the state of the
oracle machine is ``oracle origin'', then the oracle machine writes an oracle
query $x_1$ which is $n$ bits string to the oracle tape. Then, the state of
the oracle machine is changed to ``oracle appeared'', and the content of the
oracle tape is replaced by the oracle reply $y_1 = f(x_1)$. This process
repeats $m$ times. The 1-bit output of the oracle machine is calculated from
the values $<x_1, y_1>, <x_2, y_2>, \ldots, <x_m, y_m>$.
\end{definition}

Because the oracle machine determines the equality of the two distributions,
the pseudo-randomness is determined by the computational power of the oracle
machine. An effective algorithm for determining the equality of two
distributions is an oracle machine that calculates the output value within a
probabilistic polynomial-time. Thus, the pseudo-random distribution refer to
a distribution that can not be distinguished from the true random
distribution using probabilistic polynomial-time oracle machines. At this
time, the indistinguishability is defined as the following.

\begin{definition}[Polynomial-Time Indistinguishability] \label{def:ind}
If two distributions $X$ and $Y$ are indistinguishable in polynomial-time,
then the following equation holds for all probabilistic polynomial-time
algorithms $D$, all polynomials $p(\cdot)$, and sufficiently large $n$ values
    ​​$$| \Pr (D(X,1^n) = 1) - \Pr (D(Y,1^n) = 1) | <1/p(n)$$
where the output of algorithm $D$ is 1 bit.
\end{definition}

A cryptographically secure pseudo-random bit generator was first introduced
by Blum and Micali \cite{BlumM84}. After that, a number of pseudo-random bit
generators have been proposed based on number theoretic problems
\cite{BlumBS86,Kaliski87,VaziraniV85}. H{\aa}stard et al. have shown that a
pseudo-random bit generator can be constructed using an arbitrary one-way
function \cite{HastadILL99}. A pseudo random bit generator is defined as
follows. In this case, an ideal random bit generator $R$ has a uniform
distribution of all possible output values.

\begin{definition}[Pseudo-Random Bit Generator]
A pseudo-random bit generator is defined by a deterministic polynomial-time
algorithm $G$ and satisfies the following two conditions:
\begin{enumerate}
\item Scalability: $|G(s)| > |s|$  for all $s \in \bits^*$.
\item Pseudo-Randomness: The algorithm $G$ is indistinguishable from the
    ideal random bit generator $R$ in polynomial time.
\end{enumerate}
\end{definition}

Blum and Micali showed that it is possible to construct a pseudo-random bit
generator using the difficulties of the discrete logarithm problem
\cite{BlumM84}. The pseudo-random bit generator of Blum and Micali is
described in Example \ref{exa:bm_prbg}.

\begin{example}[Blum-Micali Pseudo-Random Bit Generator] \label{exa:bm_prbg}
Let $p$ be a large prime, and $g$ be a generator of $\Z_p^*$. The set $D$ is
defined as $D = \Z_p^* = {0, 1, \ldots, p-1}$. The function $f:D \rightarrow
D$ is defined as $f(x) = g^x \mod p$. The function $B:D \rightarrow \bits$ is
defined as $B(x)=1$ for $0 \leq \log_g(x) \leq (p-1)/2$ and $B(x)=0$ for
$\log_g(x) > (p-1)/2$. The Blum-Micali pseudo-random bit generation algorithm
is described as follows.
\vs \\
\indent \quad \textbf{generate} a large prime $p$ and a generator $g$ of $\Z_p^*$. \\
\indent \quad \textbf{select} a random integer $x_0$ from $D = \{ 0, 1, \ldots, p-1 \}$. \\
\indent \quad \textbf{for} $1 \leq i \leq \ell$ \textbf{do} \\
\indent \qquad      $x_i \leftarrow f(x_{i-1})$. \\
\indent \qquad      $b_i \leftarrow B(x_i)$. \\
\indent \quad \textbf{end for} \\
\indent \quad \textbf{output} $b_1, b_2, \ldots, b_\ell$.
\end{example}

Blum, Blum, and Shub showed that it is possible to construct a pseudo-random
bit generator using the difficulties of the quadratic residuacity problem
\cite{BlumBS86}. The pseudo-random bit generator of Blum, Blum, and Shub is
described in Example \ref{exa:bbs_prbg}.

\begin{example}[Blum-Blum-Shub Pseudo-Random Bit Generator] \label{exa:bbs_prbg}
Let $LSB(x)$ be a function that outputs the least significant bit of a binary
string $x$. The Blum-Blum-Shub pseudo-random bit generation algorithm is
described as follows.
\vs \\
\indent \quad \textbf{generate} a large prime $p$ such that $p \mod 4 = 3$. \\
\indent \quad \textbf{generate} a large prime $q$ such that $q \mod 4 = 3$. \\
\indent \quad $n \leftarrow pq$. \\
\indent \quad \textbf{select} a random integer $s \in \{ 1, \ldots, n-1 \}$
               such that $gcd(s,n) = 1$. \\
\indent \quad $x_0 \leftarrow s^2 \mod n$. \\
\indent \quad \textbf{for} $1 \leq i \leq \ell$ \textbf{do} \\
\indent \qquad      $x_i \leftarrow x_{i-1}^2 \mod n$. \\
\indent \qquad      $b_i \leftarrow LSB(x_i)$. \\
\indent \quad \textbf{end for} \\
\indent \quad \textbf{output} $b_1, b_2, \ldots, b_\ell$.
\end{example}

A pseudo-random function generator is a function generator that is
indistinguishable from an ideal random function generator $H$ that generates
all possible functions with a uniform probability distribution. Goldreich,
Goldwasser, and Micali showed that a pseudo-random bit generator can be used
to create a pseudo-random function generator \cite{GoldreichGM86}. A
pseudo-random function generator is defined as follows.

\begin{definition}[Pseudo-Random Function Generator] \label{def:prf}
A pseudo-random function generator is defined as an algorithm $F$ that
generates a set of functions. For all probabilistic polynomial-time oracle
$M$, all polynomials $p(\cdot)$, and sufficiently large $n$ values, it
satisfies
    $$| \Pr(M^F (1^n) = 1) - \Pr(M^H (1^n) = 1) | < 1/p(n)$$
where $H$ is an ideal random function generator that generates all possible
functions as a uniform probability distribution.
\end{definition}

Goldreich, Goldwasser and Micali proposed a pseudo-random function generator
as follows.

\begin{example}[GGM Pseudo-Random Function Generator] \label{exm:ggm-prf}
Let $G$ be a pseudo random bit generator whose input is $k$ bits and whose
output is $2k$ bits. That is, $G(x) = b_1^x b_2^x \cdots b_{2k}^x$ for the
initial value $x \in I_k$. Let $G_0(x)$ be the first $k$ bit string of $G(x)$
and $G_1(x)$ be the remaining $k$ bit string of $G(x)$. That is, $G_0(x) =
b_1^x \cdots b_k^x$ and $G_1(x) = b_{k+1}^x \cdots b_{2k}^x$. For a $t$-bit
binary string $\alpha = \alpha_1 \alpha_2 \cdots \alpha_t$, $G_{\alpha}(x)$
is defined as $G_{\alpha_t} (\cdots (G_{\alpha_2} (G_{\alpha_1} (x))))$. For
a given $x \in I_k$, a pseudo-random function $f_x: I_k \rightarrow I_k$ is
defined as $$f_x(y) \stackrel{def}{=} G_y(x)$$ where $y \in I_k$. For
polynomials $P_1$ and $P_2$ and given $x \in I_k$, a pseudo-random function
$F : I_{P_1(k)} \rightarrow l_{P_2(k)}$ is defined as $$F_x(y)
\stackrel{def}{=} G'(G_y(x))$$ where $y \in I_{P_1(k)}$ and $G'$ is a
pseudo-random bit generator with $k$ bits input and $P_2(k)$ bits output.
\end{example}

Let us look at the performance of the GGM pseudo-random function generator. A
pseudo-random function generator $F_x : I_{P_1(k)} \rightarrow I_{P_2(k)}$
must generate a pseudo-random string of approximately $P_1(k) \cdot 2k +
P_2(k)$ bits. Assuming that the size $k$ of the string $x$ is properly
selected and that $P_2(k)$ is not a value much larger than $P_1(k)$, then the
size of pseudo-random bits that the pseudo-random function $F_x$ must
generate is proportional to the input bit size $P_1(k)$. So the smaller the
input bit size of the function, the faster the function can be generated.

A pseudo-random permutation generator was introduced by Luby and Rackoff
\cite{LubyR88}. The definition of pseudo-random permutation generator is
given as follows.

\begin{definition}[Pseudo-Random Permutation Generator] \label{def:prp}
A pseudo-random permutation generator is defined as an algorithm $P$ which
generates a set of permutations. For all probabilistic polynomial-time oracle
$M$, all polynomials $p(\cdot)$, and sufficiently large $n$ values, it
satisfies the following equation
    $$| \Pr(M^P(1^n) = 1) - \Pr(M^H(1^n) = 1) | < 1/p(n)$$
where $H$ is an ideal random permutation generator that produces all possible
permutations with a uniform probability distribution.
\end{definition}

Luby and Rackoff showed that a pseudo-random permutation generator could be
constructed by using a pseudo-random function and a three rounds balanced
Feistel network structure.

\begin{example}
A three rounds balanced Feistel network that uses pseudo-random functions
$f_1, f_2, f_3$ is a pseudo-random permutation generator and is defined as
    $$D_{f_3} \circ D_{f_2} \circ D_{f_1} (L \| R)
    \stackrel{def}{=} D_{f_3} (D_{f_2} (D_{f_1} (L \| R)))$$
where $D_{f_i} (L \| R) = (R \| L \oplus f_i(R))$ and $|L| = |R|$.
\end{example}

A super pseudo-random permutation generator is a permutation generator, which
can not be distinguished from an ideal random permutation generator when an
oracle machine is able to access both a permutation and the inverse of the
permutation. The definition is given as follows.

\begin{definition}[Super Pseudo-Random Permutation Generator]
\label{def:sprp} A super pseudo-random permutation generator is defined as an
algorithm $P$ that generates the set of permutations. For all probabilistic
polynomial-time oracle $M$, all polynomials $p(\cdot)$, and the large $n$
value, it satisfies the following equation
    $$| \Pr (M^{P, P^{-1}} (1^n) = 1) - \Pr (M^{H, H^{-1}} (1^n) = 1) |
    < 1/p(n) |$$
where $P^{-1}$ is the inverse of the permutation $P$, $H$ is an ideal random
permutation generator that produces a permutation with a uniform probability
distribution, and $H^{-1}$ is the inverse of $H$.
\end{definition}

Luby and Rackoff show that it is possible to build a super pseudo-random
permutation generator by using a pseudo-random function and a four rounds
balanced Feistel network structure.

\begin{example}
A four rounds balanced Feistel network that uses pseudo-random functions
$f_1, f_2, f_3, f_4$ is a super pseudo-random permutation generator and is
defined as
    $$D_{f_4} \circ D_{f_3} \circ D_{f_2} \circ D_{f_1} (L \| R)
    \stackrel{def}{=} D_{f_4} (D_{f_3} (D_{f_2} (D_{f_1} (L \| R))))$$
where $D_{f_i} (L \| R) = (R \| L \oplus f_i(R))$ and $|L| = |R|$.
\end{example}

If a block cipher is a pseudo-random permutation generator, then it becomes a
secure block cipher for a chosen-plaintext attack. If a block cipher is a
super pseudo-random permutation generator, then it becomes a secure block
cipher for a chosen-plaintext attack and a chosen-ciphertext attack.

\section{Pseudo-Random Permutation Generator}

The research on pseudo-random permutation generators has been started by Luby
and Rackoff. They have proved that a three rounds balanced Feistel network
that uses pseudo-random functions is a pseudo-random permutation generator
\cite{LubyR88}. The proof of pseudo-random permutation is largely divided
into two parts. First, the proof show that a three rounds Feistel network
that uses ideal random functions becomes a pseudo-random permutation
generator. Next, the proof show that that a three rounds Feistel network that
uses pseudo-random functions is pseudo-random by using contradiction. The
proof that shows the pseudo-randomness of the three rounds Feistel network
that uses ideal random functions looks like this. First, we define $BAD$ as
an event that a machine can distinguish a balanced Feistel permutation
generator $P$ from an ideal permutation generator $K$. Next the proof show
that $P$ is the same as $K$ if the $BAD$ event does not occur, and it also
show that the probability of the $BAD$ event is very low. To prove that a
three rounds Feistel network that uses pseudo-random functions is
pseudo-random by using contradiction. In other words, if the three rounds
Feistel network that uses ideal random functions is pseudo-random, but the
three rounds Feistel network that uses pseudo-random functions is not
pseudo-random, then it is possible to derive a contradiction on the
assumption that pseudo-random functions are pseudo-random.

There have been some studies to simply prove the pseudo-randomness of
permutation generators since the work of Luby and Rackoff. Maurer used a
local random function instead of a pseudo-random function to show that a
three rounds Feistel network is a pseudo-random permutation generator
\cite{Maurer92}. Naor and Reingold have proved that a three rounds structure
that uses a pairwise independent permutation and a two rounds Feistel network
is pseudo-random \cite{NaorR99}.

After the work of Luby and Rackoff, much research focused to build
pseudo-random permutation generators by using balanced Feistel networks and
pseudo-random functions with small number of rounds. That is, if we use
pseudo-random functions with small number of rounds, then the size of keys
used in the permutation can be decreased since the size of keys for
pseudo-random functions is large. Let $f, g, h, e$ be pseudo-random
functions, and $f^i$ be the composition of the function $f$ with $i$ times.
The results are summarized as follows.
\begin{itemize}
\item $D_f \circ D_f \circ D_f$ and $D_f \circ D_g \circ D_f$ are not
    pseudo-random permutations \cite{Rueppel90}.

\item $D_g \circ D_g \circ D_f$ and $D_g \circ D_f \circ D_f$ are not
    pseudo-random permutations \cite{Ohnishi88,ZhengMI90o}.

\item $D_g \circ D_g \circ D_f \circ D_f$ is not a super pseudo-random
    permutation \cite{SadeghiyanP92}.

\item For all $i, j, k \geq 1$, $D_{f^k} \circ D_{f^j} \circ D_{f^i}$ is
    not a pseudo-random permutation \cite{ZhengMI90i}.

\item For all $i, j, k, \ell \geq 1$, $D_{f^\ell} \circ D_{f^k} \circ
    D_{f^j} \circ D_{f^i}$ is not a super pseudo-random permutation
    \cite{SadeghiyanP92}.

\item $D_{f^2} \circ D_{f} \circ D_{f} \circ D_{f}$ is a pseudo-random
    permutation \cite{Pieprzyk91}.

\item If $I$ is an identity function, $D_{f^2} \circ D_{I} \circ D_{f}
    \circ D_{f^2} \circ D_{I} \circ D_{f}$ is a super pseudo-random
    permutation \cite{SadeghiyanP92}.

\item If $\zeta$ is a simple function like a shift operation, $D_{f \circ
    \zeta \circ f} \circ D_{f} \circ D_{f} \circ D_{f}$ is a pseudo-random
    permutation and $D_{f \circ \zeta \circ f} \circ D_{f} \circ D_{f}
    \circ D_{f} \circ D_{f}$ is a super pseudo-random permutation
    \cite{Patarin92h}.
\end{itemize}

Although many studies are concerned with the use of a small number of
pseudo-random functions, reducing the size of the secret key using fewer
pseudo-random functions is not a huge benefit. This is because it is possible
to increase a small key to a large key using a pseudo random bit generator.

\begin{figure}[t]
\centering
\includegraphics[scale=0.75]{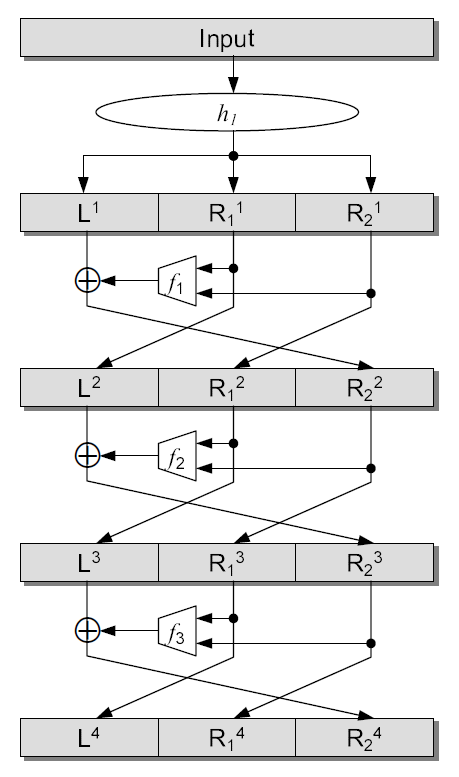}
\caption{Naor and Reingold's PRP}
\label{fig:nr_prp}
\end{figure}

\begin{figure}[t]
\centering
\includegraphics[scale=0.75]{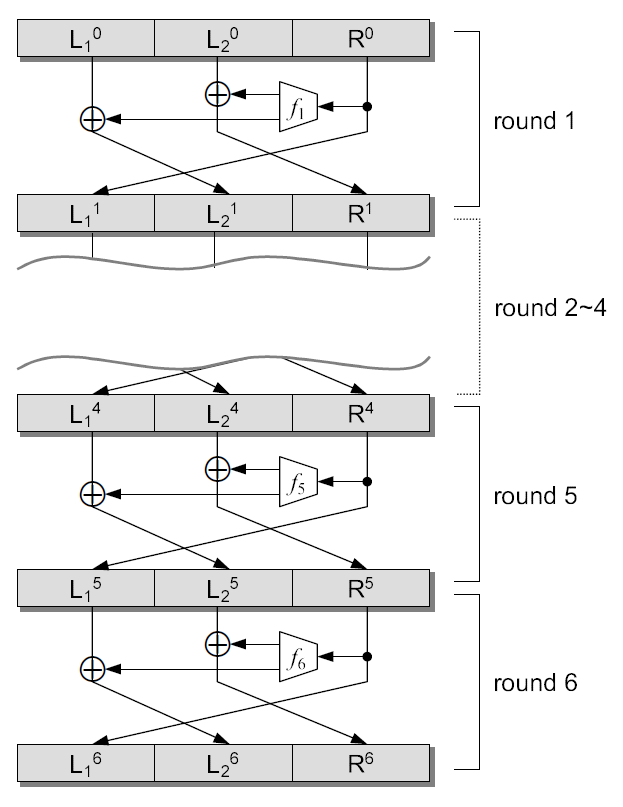}
\caption{Jutla's PRP}
\label{fig:jutla_prp}
\end{figure}

Research to build pseudo-random permutation generators from unbalanced
Feistel network structures has only recently begun \cite{NaorR99,Jutla98}.
Naor and Reingold showed that an unbalanced Feistel network with total $k+2$
number of rounds is a pseudo-random permutation generator if it uses a
pairwise independent permutation in the first round, a source-heavy
unbalanced Feistel network where the size of a source block is $k$ times
larger than a target-block in the remaining rounds, and the $F$ function of
the Feistel network is a pseudo-random function \cite{NaorR99}. The pairwise
independent permutation is defined as a permutation in which the distribution
of the function output values ​​is uniform even if any two input values ​​are
selected. The pseudo-random permutation generator of Naor and Reingold is
given in Figure \ref{fig:nr_prp} where $k$ is two, $h_1$ is a pairwise
independent of permutation, and $f_1, f_2$, and $f_3$ are pseudo-random
functions.

Jutla showed that an unbalanced Feistel network with total $2k+2$ number of
rounds is a pseudo-random permutation generator if it uses a pseudo-random
function for the $F$-function of the Feistel network, and a target-heavy
unbalanced Feistel network where a target-block size is larger than a source
block by $k$ times \cite{Jutla98}. At this time, the probability that an
oracle machine that distinguishes between an ideal random permutation
generator and a $2k+2$ rounds target-heavy unbalanced Feistel network is less
than $(m^k / 2^{kn})$. The pseudo-random permutation generator of Jutla is
given in Figure \ref{fig:jutla_prp} where $k$ is two and total rounds is six.

\chapter{Analysis of Unbalanced Feistel Networks} \label{chap:analysis-ufn}

In this chapter, we analyze the conditions for permutation generators based
on Feistel networks to be pseudo-random. This chapter is summarized as
follows. An unbalanced Feistel network is a Feistel network with different
sizes of source and target blocks. The unbalanced Feistel network is largely
divided into a source-heavy unbalanced Feistel network and a target-heavy
unbalanced Feistel network. For a source-heavy unbalanced Feistel network
($kn$:$n$-UFN) where a source block is $k$ times larger than a target block,
a $k+2$ rounds $kn$:$n$-UFN using pseudo-random functions is a pseudo-random
permutation generator. For a target-heavy unbalanced Feistel network
($n$:$kn$-UFN) where a target-block is $k$ times larger than a source block,
a $k+2$ rounds $n$:$kn$-UFN using pseudo-random functions is a pseudo-random
permutation generator. Therefore, the minimum number of rounds for a
unbalanced Feistel network using pseudo-random functions to be pseudo-random
is $k+2$ rounds.

The structure of this chapter is as follows. In Section 3.1, we divide
unbalanced Feistel networks into two categories. In Section 3.2, we overview
the proof method to prove the pseudorandomness of an unbalanced Feistel
network. In Section 3.3, we analyze the conditions for a source-heavy
unbalanced Feistel network to be pseudo-random. In Section 3.4, we analyze
the conditions for a target-heavy unbalanced Feistel network to be
pseudo-random.

\section{Definition and Category}

In a Feistel network, a block that is the input of a round function is called
a source block, and a block that is combined with the output of a round
function is called a target block. An unbalanced Feistel network is a Feistel
network with different source and target block sizes. An unbalanced Feistel
network with a source-block size of $s$ bits and a target block size of $t$
bits is denoted as $s$:$t$-UFN.

\begin{figure}
\centering
\includegraphics[scale=0.75]{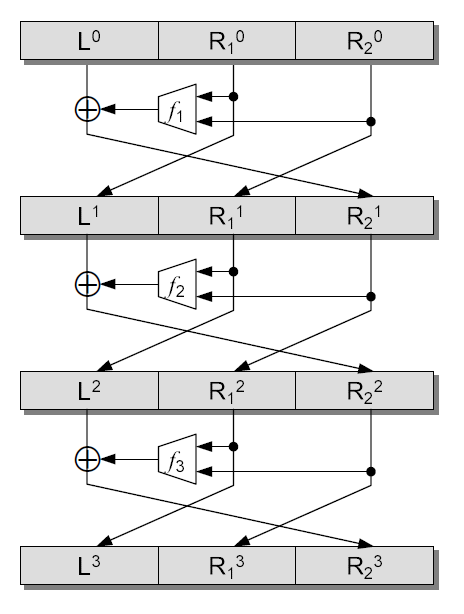}
\caption{The structure of a 3 rounds $2n$:$n$-UFN}
\label{fig:3r_2nn_ufn}
\end{figure}

An unbalanced Feistel network is largely classified as a source-heavy
unbalanced Feistel network or a target-heavy unbalanced Feistel network.
A source-heavy unbalanced Feistel network is an unbalanced Feistel network
where the size of a source block is larger than that of a target block. The
source-heavy unbalanced Feistel network is denoted by $kn$:$n$-UFN and it is
defined as follows. For instance, a 3 rounds $2n$:$n$-UFN structure is
described in Figure \ref{fig:3r_2nn_ufn}.

\begin{definition}[$kn$:$n$-UFN] \label{def:knn-ufn}
For any function $f$ belonging to the set of functions $F:I_{kn} \rightarrow
I_n$, one round $kn$:$n$-UFN is defined by the following permutation
    $$D_f (L \| R_1 \| \cdots \| R_k)
    \stackrel{def}{=} (R_1 \| \cdots \| R_k \| L \oplus f(R_1 \| \cdots R_k)).$$
Similarly, for any functions $f_1, f_2, \ldots, f_r$ belonging to the set of
functions $F:I_{kn} \rightarrow I_n$, an $r$ rounds $kn$:$n$-UFN is defined
by the following permutation
    $$D_{f_r} \circ \cdots \circ D_{f_2} \circ D_{f_1}
        (L^0 \| R_1^0 \| \cdots \| R_k^0)
    \stackrel{def}{=} D_{f_r} \circ \cdots \circ D_{f_2} ( D_{f_1}
        (L^0 \| R_1^0 \| \cdots \| R_k^0) ).$$
In this case, we have $|L| = |R_i| = n$.
\end{definition}

\begin{figure}
\centering
\includegraphics[scale=0.75]{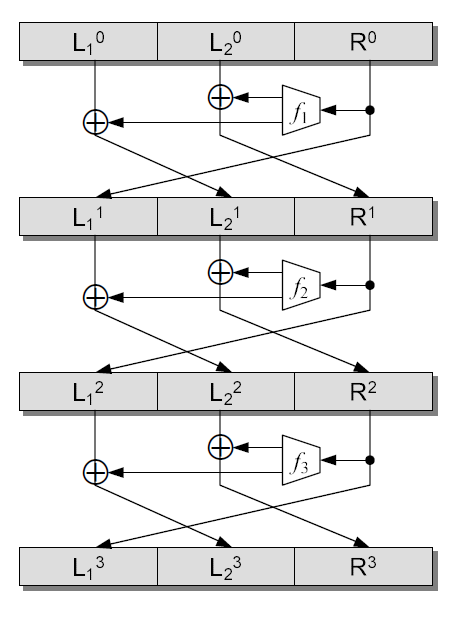}
\caption{The structure of a 3 rounds $n$:$2n$-UFN}
\label{fig:3r_n2n_ufn}
\end{figure}

A target-heavy unbalanced Feistel network is an unbalanced Feistel network
where the size of a target block is larger than that of a source block. The
target-heavy unbalanced Feistel network is denoted by $n$:$kn$-UFN and it is
defined as follows. For instance, a 3 rounds $n$:$2n$-UFN structure is
described in Figure \ref{fig:3r_n2n_ufn}.

\begin{definition}[$n$:$kn$-UFN]
For any function $f$ belonging to the set of functions $F:I_{n} \rightarrow
I_{kn}$, one round $n$:$kn$-UFN is defined by the following permutation
    $$D_f (L_1 \| \cdots \| L_k \| R)
    \stackrel{def}{=} (R \| ( L_1 \| \cdots \| L_k ) \oplus f(R))
    = (R \| L_1 \oplus C_1(f(R)) \| \cdots \| C_k(f(R))).$$
In this case, the function $C(\cdot)$ satisfies $C_1(f(R)) \| \cdots \|
C_k(f(R)) = f(R)$. Similarly, for any functions $f_1, f_2, \ldots, f_r$
belonging to the set of functions $F:I_{n} \rightarrow I_{kn}$, an $r$ rounds
$n$:$kn$-UFN is defined by the following permutation
    $$D_{f_r} \circ \cdots \circ D_{f_2} \circ D_{f_1}
        (L_1^0 \| \cdots \| L_k^0 \| R^0)
    \stackrel{def}{=}
    D_{f_r} \circ \cdots \circ D_{f_2} ( D_{f_1}
        (L_1^0 \| \cdots \| L_k^0 \| R^0) ).$$
In this case, we have $|L_i| = |R| = |C_i(\cdot)| = n$.
\end{definition}

\section{Overview of the Pseudo-Random Proof}

We first overview the way to prove that an $r$ rounds unbalanced Feistel
network using a pseudo-random function generator is a pseudo-random
permutation generator. The overall method is similar to the method used by
Luby and Rackoff \cite{LubyR88}.

First, we show that an $r-1$ rounds unbalanced Feistel network is not a
pseudo-random permutation generator. For this, we show that there exists a
linear relationship between the input and output values ​​of the $r-1$ rounds
unbalanced Feistel network. Then we use this linear relationship to build an
oracle machine that distinguishes between an ideal random permutation
generator and the $r-1$ rounds unbalanced Feistel network permutation
generator.

Next, we show that if an $r$ rounds unbalanced Feistel network that uses
ideal random functions is pseudo-random, then an $r$ rounds unbalanced
Feistel network that uses pseudo-random functions is also pseudo-random. This
is because if the $r$ rounds unbalanced Feistel network using ideal random
functions is pseudo-random but the $r$ rounds unbalanced Feistel network
using pseudo-random functions is not pseudo-random, then it is possible to
show an contradiction that the pseudo-random function is pseudo-random.

The following is a proof strategy to showing that an $r$ rounds unbalanced
Feistel network using ideal random functions is pseudo-random. We first
define the case where the $r$ rounds unbalanced Feistel network is not a
pseudo-random permutation generator as a BAD event. For the BAD event, we
prove the following two things.
\begin{enumerate}
\item If the BAD event does not occur, the output of the $r$ rounds
    unbalanced Feistel network that uses ideal random functions is uniform.

\item The probability of the BAD event is very low.
\end{enumerate}
By using these two things, we can prove that the $r$ rounds Feistel network
using ideal random functions is a pseudo-random permutation generator.

\section{The Pseudo-Random Proof of $kn$:$n$-UFN} \label{sec:analysis-knn-ufn}

The following theorem show that a $kn$:$n$-UFN is not pseudo-random if the
number of rounds is less than or equal to $k+1$.

\begin{theorem} \label{thm:k+1_knn_ufn_not_prp}
A $k+1$ rounds $kn$:$n$-UFN is not pseudo-random.
\end{theorem}

\begin{proof}
For the proof, we first show that there is a linear relationship between the
input and output values ​​of the $kn$:$n$-UFN, and that this linear
relationship can be used to create an oracle machine that can distinguish
between an ideal random permutation generator and the $kn$:$n$-UFN. From the
definition of a $kn$:$n$-UFN, we first obtains the following equation
    $$L^{k+1} = R_1^k = R_2^{k-1} = \cdots = R_k^1
    = L^0 \oplus f(R_1 \| \cdots \| R_k).$$

We select two oracle queries $x_p = (L_p^0 \| R_{p,1}^0 \| \cdots \|
R_{p,k}^0 )$ and $x_q = ( L_q^0 \| R_{q,1}^0 \| \cdots \| R_{q,k}^) )$ where
$p$ and $q$ are indexes of two oracle queries with $1 \leq p < q \leq m$.
Then, we have $L_p^{k+1} = L_p^0 \oplus f_1(R_{p,1}^0 \| \cdots \| R_{p,k}^0
)$ and $L_q^{k+1} = L_q^0 \oplus f_1(R_{q,1}^0 \| \cdots \| R_{q,k}^0 )$.
Therefore, if we choose two oracle queries $x_p$ and $x_q$ in which $L_p^0$
and $L_q^0$ are only different, then we can derive the following relation
    $$L_p^{k+1} \oplus L_q^{k+1} = L_p^0 \oplus L_q^0$$
since $f_1(R_{p,1}^0 \| \cdots \| R_{p,k}^0 ) = f_1(R_{q,1}^0 \| \cdots \|
R_{q,k}^0 )$.

This linear relation can be used to build an oracle machine $M$ that
distinguishes between the $kn$:$n$-UFN and an ideal random permutation
generator. First, the oracle machine creates two oracle queries $x_p$ and
$x_q$ that differ only in $L^0$ values and receives the responses $<x_p,
y_p>, <x_q, y_q>$. If the equation $L_p^{k+1} \oplus L_q^{k+1} = L_p^0 \oplus
L_q^0$ is satisfied from the responses of the queries, then the oracle
machine outputs $1$. Otherwise, the oracle machine outputs $0$. Thus, if the
$x_p$ and $x_q$ values ​​are generated by the $kn$:$n$-UFN, then the output of
the oracle machine is always $1$. However, if the $x_p$ and $x_q$ values are
generated by the ideal random permutation generator, then the probability
that the equation is satisfied is $1 / 2^n$. Therefore we have the following
equation
    $$\big| \Pr (M^P(1^{(k+1)n}) = 1) - \Pr (M^K(1^{(k+1)n} = 1) \big|
    = 1 - 1/2^n > 1/p(n)$$
where $P$ is the $kn$:$n$-UFN and $K$ is the ideal random permutation
generator. From this equation, the $k+1$ rounds $kn$:$n$-UFN is not
pseudo-random.
\end{proof}

We now prove that a $k+2$ rounds $kn$:$n$-UFN permutation generator using
pseudo-random functions is pseudo-random. First, we define an event that can
be used to distinguish a $k+2$ rounds $kn$:$n$-UFN using ideal random
functions from an ideal random permutation generator as the following BAD
event.

\begin{definition}(The BAD event $\xi$ of a $k+2$ rounds $kn$:$n$-UFN)
\label{def:bad-knn-ufn} %
A random variable $\xi^i$ is defined as an event in which $R_{p,1}^i \|
\cdots \| R_{p,k}^i$ and $R_{q,1}^i \| \cdots \| R_{q,k}^i$ of two oracle
queries with indexes $p$ and $q$ are equal where $1 \leq p < q \leq m$. The
BAD event $\xi$ is a random variable defined as $\vee_{i=1}^{k+1} \xi^i$.
\end{definition}

If an oracle machine is able to distinguish between a $k+2$ rounds
$kn$:$n$-UFN and an ideal random permutation generator, then the BAD event
must occur. This means that if the BAD event does not occur, then the $k+2$
rounds $kn$:$n$-UFN is equal to the ideal random permutation generator. In
the following lemma, we prove it.

\begin{lemma} \label{lem:knn_ufn_not_bad}
A $k+2$ rounds $kn$:$n$-UFN permutation generator using an ideal random
function generator is equal to an ideal random permutation generator if the
BAD event does not occur. That is, for all possible $\sigma_1, \sigma_2,
\ldots, \sigma_m \in \bits^{(k+1)n}$, we have
    $$\Pr (\wedge_{i=1}^m (y_i = \sigma_i) | \neg\xi) = 1/2^{(k+1)nm}$$
where $y_i$ is the output of the $k+2$ rounds $kn$:$n$-UFN permutation
generator.
\end{lemma}

\begin{proof}
From the definition of a $k+2$ rounds $kn$:$n$-UFN using ideal random
functions, the reply $y_p$ of the $p$th oracle machine query $x_p$ is
described as
    \begin{align*}
    y_p
    &= (L_p^{k+2} \| R_{p,1}^{k+2} \| \cdots \| R_{p,k}^{k+2}) \\
    &= (L_p^1 \oplus h_2(\cdot) \| L_p^2 \oplus h_3(\cdot) \| \cdots \|
        L_p^{k+1} \oplus h_{k+2}(\cdot))
    \end{align*}
where $h_1, h_2, \ldots, h_{k+2}$ are functions generated by an ideal random
function generator whose inputs are $kn$ bits and whose outputs are $n$ bits,
and the input values ​​of $h_i(\cdot)$ are $(R_{p,1}^{i-1} \| \cdots \|
R_{p,k}^{i-1})$. By the definition of the BAD event $\xi$, we obtain the
following equation
    $$\neg \xi = \wedge_{i=1}^{k+1} \neg \xi^i
    = \wedge_{i=1}^{k+1} (R_{p,1}^{i} \| \cdots \| R_{p,k}^{i} \neq R_{q,1}^{i}
      \| \cdots \| R_{q,k}^{i})$$
where $p$ and $q$ are the oracle indexes with $1 \leq p < q \leq m$.

Thus the output value of $h$ becomes a value with uniform distribution since
the input values ​​of $h_2, \ldots, h_{k+2}$ are different for all oracle
queries and $h$ is a function generated by the ideal random function
generator. Therefore, the value of $L^{i-1} \oplus h_i(\cdot)$ becomes a
value with uniform distribution.
\end{proof}

\begin{lemma} \label{lem:knn_ufn_bad}
The probability of the BAD event in a $k+2$ rounds $kn$:$n$-UFN permutation
generator is bounded by
    $$\Pr(\xi) \leq (k+1) m^2 / 2^{n+1}.$$
\end{lemma}

\begin{proof}
By the definition of the BAD event, we have $\xi = \vee_{i=1}^{k+1} \xi^i$.
We first calculate the probability of each event $\xi^i$. A random variable
$\xi^i$ represents an event that $R_{p,1}^i \| \cdots \| R_{p,k}^i$ and
$R_{q,1}^i \| \cdots \| R_{q,k}^i$ are equal for the indexes $p$ and $q$ of
the oracle queries with $1 \leq p < q \leq m$. By the definition of a
$kn$:$n$-UFN structure, we obtain the following equation
    \begin{align*}
    & R_1^i \| \cdots \| R_{k-i}^i \| R_{k-i+1}^i \| \cdots \| R_k^i \\
    &= R_{i+1}^0 \| \cdots \| R_{k}^0 \| L^0 \oplus h_1(\cdot) \| \cdots
      \| L^{i-1} \oplus h_i(\cdot)
    \end{align*}
where each $h_i(\cdot)$ is a function generated by an ideal random function
generator.

Thus, the best choice for the event $\xi^i$ to occur is to select an oracle
query with $R_{p,j}^0 = R_{q,j}^0$ for $i+1 < j < k$. In this case, the
probability of the event $\xi^i$ is $\Pr (\xi^i) = {}_mC_2 \cdot 1 / 2^{in}$.
Therefore, we have $\Pr (\xi) < (k +1) m^2 / 2^{n+1}$.
\end{proof}

In Theorem \ref{thm:k+2_knn_ufn_prp}, we prove that a $k+2$ rounds
$kn$:$n$-UFN that uses ideal random functions is pseudo-random from the above
two lemmas, and we also prove that a $k+2$ rounds $kn$:$n$-UFN that uses
pseudo-random functions is also pseudo-random.

\begin{theorem} \label{thm:k+2_knn_ufn_prp}
A $k+2$ rounds $kn$:$n$-UFN permutation generator using pseudo-random
functions is a pseudo-random permutation generator.
\end{theorem}

\begin{proof}
From Lemma \ref{lem:knn_ufn_not_bad} and Lemma \ref{lem:knn_ufn_bad}, we can
show that a $kn$:$n$-UFN permutation generator $P$ using an ideal random
function generator is a pseudo-random permutation generator. First, we have
$| \Pr (M^P(1^{(k+1)n}) = 1 | \neg \xi) - \Pr (M^K (1^{(k+1)n}) = 1) | = 0$
since $P$ is the same as an ideal random permutation generator $K$ when the
BAD event $\xi$ does not occur by Lemma \ref{lem:knn_ufn_not_bad}. We also
have $| \Pr (M^P(1^{(k+1)n}) = 1 | \xi) - \Pr (M^K (1^{(k+1)n}) = 1) | \leq
1$ since the absolute value of the probability difference is less than 1.
Therefore, we have the following equation
    \begin{align*}
    &   \big| \Pr (M^P(1^{(k+1)n}) = 1) - \Pr (M^K(1^{(k+1)n}) = 1) \big| \\
    &=  \big| \Pr (M^P(1^{(k+1)n}) = 1 | \xi) - \Pr (M^K(1^{(k+1)n}) = 1) \big|
        \cdot \Pr(\xi) + \\
    &\quad~
        \big| \Pr (M^P(1^{(k+1)n}) = 1 | \neg \xi) - \Pr (M^K(1^{(k+1)n}) = 1) \big|
        \cdot \Pr(\neg \xi) \\
    &\leq \Pr(\xi)
     \leq (k+1)m^2 / 2^{n+1}.
    \end{align*}

We now show that a $k+2$ rounds $kn$:$n$-UFN permutation generator $P$ using
a pseudo-random function generator is a pseudo-random permutation generator.
For this, we use the proof by contradiction. That is, if a $k+2$ rounds
$kn$:$n$-UFN using an ideal random function generator is pseudo-random but a
$k+2$ rounds $kn$:$n$-UFN using a pseudo-random function generator is not
pseudo-random, then we can derive a contradiction to the pseudo-randomness of
the pseudo-random function generator.

Suppose that a $kn$:$n$-UFN permutation generator $P$ using a pseudo-random
function generator is not pseudo-random. Then there exists an oracle machine
$M$ which distinguishes an ideal random permutation generator $K$ and $P$
with a probability greater than $1/n^c$ for a constant $c$.

First, we let $D_{f_{k+2}} \circ \cdots \circ D_{f_{i+1}} \circ D_{h_i} \circ
\cdots \circ D_{h_1}(\cdot)$ be a permutation generator in which an ideal
random function generator is used from the first round to the $i$th round and
a pseudo-random function generator is used from the $i+1$th round to the
$k+2$th round in a $kn$:$n$-UFN permutation generator for $i$ with $0 \leq i
\leq k + 2$. Let $p_i^D$ be the probability that an oracle machine that has
access to this permutation generator will output $1$. That is,
    \begin{align*}
    p_i^D
    =   \Pr (M^{D_{f_{k+2}} \circ \cdots \circ D_{f_{i+1}} \circ D_{h_i}
        \circ \cdots \circ D_{h_1}} (1^{(k+1)n}) = 1)
    \end{align*}
where $f_{i+1}, \ldots, f_{k+2}$ are functions generated by a pseudo-random
function generator and $h_1, \ldots, h_i$ are functions generated by an ideal
random function generator. Let $p^K$ be the probability that an oracle
machine that has access to the ideal random permutation generator will output
$1$. Since we supposed that the $k+2$ rounds $kn$:$n$-UFN using a
pseudo-random function generator is not pseudo-random, we get the following
equation
    \begin{align*}
    1/n^c
    &\leq | p^K - p_0^D | \\
    &\leq | p^K - p_{k+2}^D | + | p_{k+2}^D - p_{k+1}^D | + \cdots +
          | p_{i+1}^D - p_{i}^D | + \cdots + | p_1^D - p_0^D |.
    \end{align*}

However, since the $k+2$ rounds $kn$:$n$-UFN using an ideal random function
generator has been shown to be a pseudo-random permutation generator, we have
$| P^K - P_{k+2}^D | \leq (k+1) m^2 / 2^{n+1}$. Therefore, we have $|
p_{i+1}^D - p_{i}^D | \geq 1 / (k+3)n^c$ for some $i$. By using this, an
oracle machine $C$ that distinguishes an ideal random function generator and
a pseudo-random function generator with a probability higher than $1 /
(k+3)n^c$ can be constructed as follows.

In the oracle machine $C$, we first change the part of $C$ that calculates
the oracle query reply to $D_{f_{k+2}} \circ \cdots \circ D_{f_{i+1}} \circ
D_X \circ D_{h_{i-1}} \circ \cdots \circ D_{h_1} (\cdot)$. In this case, $X$
is a function whose input is $kn$ bits and whose output is $n$ bits. If $X$
in $C$ is a function generated by a pseudo-random function generator $F$,
then we have $p_{i}^D = \Pr (C^F (1^{kn}) = 1)$. On the other hand, if $X$ in
$C$ is a function generated by an ideal random function generator $H$, then
we have $p_{i+1}^D = \Pr (C^H (1^{kn}) = 1)$. However, it is possible to
distinguish the ideal random function generator and the pseudo-random
function generator with a probability greater than $1/(k+3)n^c$ since $|
p_{i+1}^D - p_{i}^D | \geq 1/(k+3)n^c$. But this contradicts that the
pseudo-random function generator is pseudo-random. Therefore, the $k+2$
rounds $kn$:$n$-UFN using a pseudo-random function generator is a
pseudo-random permutation generator.
\end{proof}

From Theorem \ref{thm:k+1_knn_ufn_not_prp} and Theorem
\ref{thm:k+2_knn_ufn_prp}, the minimum number of rounds of the $kn$:$n$-UFN
permutation generator using a pseudo-random function generator to be
pseudo-random is $k+2$.

\section{The Pseudo-Random Proof of $n$:$kn$-UFN} \label{sec:analysis-nkn-ufn}

The following theorem show that a $n$:$kn$-UFN permutation generator is not
pseudo-random if the number of rounds is less than or equal to $k+1$.

\begin{theorem} \label{thm:k+1_nkn_ufn_not_prp}
A $k+1$ rounds $n$:$kn$-UFN permutation generator is not pseudo-random.
\end{theorem}

\begin{proof}
From the definition of a $n$:$kn$-UFN, we obtain the following equation
    \begin{align*}
    L_1^{k+1}
    &=  R^k
     =  L_k^{k-1} \oplus C_k(f_k(R^{k-1})) \\
    &=  L_{k-1}^{k-2} \oplus C_{k-1}(f_{k-1}(R^{k-2})) \oplus
        C_{k}(f_{k}(R^{k-1})) \\
    &\quad \vdots \\
    &=  L_1^0 \bigoplus_{i=1}^{k} C_i(f_i(R^{i-1})).
    \end{align*}

If we choose two oracle queries $x_p$ and $x_q$ in which $L_p^0$ and $L_q^0$
are only different, then we have $R_p^i = R_q^i$ for all $i$ with $1 \leq i
\leq k-1$. Therefore, from the oracle responses $<x_p, y_p>, <x_q, y_q>$, we
can derive the following relation
    \begin{align*}
    L_{p,1}^{k+1} \oplus L_{q,1}^{k+1} = L_{p,1}^0 \oplus L_{q,1}^0.
    \end{align*}

By using this linear relation, we can build an oracle machine $M$ that
distinguishes between an $n$:$kn$-UFN permutation generator and an ideal
random permutation generator. First, the oracle machine creates two oracle
queries $x_p$ and $x_q$ that differ only in $L_1^0$ values and receives
oracle responses $<x_p, y_p>, <x_q, y_q>$. If the equation $L_{p,1}^{k+1}
\oplus L_{q,1}^{k+1} = L_{p,1}^0 \oplus L_{q,1}^0$ is satisfied from the
responses of the queries, then the oracle machine outputs $1$. Otherwise, the
oracle machine outputs $0$. Thus, if the $x_p$ and $x_q$ values ​​are generated
by the $n$:$kn$-UFN, then the output of the oracle machine is always $1$.
However, if the $x_p$ and $x_q$ values are generated by the ideal random
permutation generator, then the probability that the relation is satisfied is
$1/2^n$. Therefore we have the following equation
    \begin{align*}
    \big| \Pr (M^P(1^{(k+1)n}) = 1) - \Pr (M^K(1^{(k+1)n} = 1) \big|
    = 1 - 1/2^n > 1/p(n)
    \end{align*}
where $P$ is the $n$:$kn$-UFN and $K$ is the ideal random permutation
generator. Thus the $k+1$ rounds $n$:$kn$-UFN permutation generator is not
pseudo-random.
\end{proof}

We now prove that a $k+2$ rounds $n$:$kn$-UFN permutation generator using
pseudo-random functions is pseudo-random. First, we define an event that can
be used to distinguish a $k+2$ rounds $n$:$kn$-UFN using ideal random
functions from an ideal random permutation generator as the following BAD
event.

\begin{definition}(The BAD event $\xi$ of a $k+2$ rounds $n$:$kn$-UFN)
\label{def:bad-k+2-nkn-ufn} %
A random variable $\xi^i$ is defined as an event in which $R_{p}^i$ and
$R_{q}^i$ of two oracle queries with indexes $p$ and $q$ are equal where $1
\leq p < q \leq m$. The BAD event $\xi$ is a random variable defined as
$\vee_{i=k}^{k+1} \xi^i$.
\end{definition}

\begin{lemma} \label{lem:nkn_ufn_not_bad}
A $k+2$ rounds $n$:$kn$-UFN permutation generator using an ideal random
function generator is equal to an ideal random permutation generator if the
BAD event does not occur. That is, for all possible $\sigma_1, \sigma_2,
\ldots, \sigma_m \in \bits^{(k+1)n}$, we have
    $$\Pr (\wedge_{i=1}^m (y_i = \sigma_i) | \neg \xi) = 1/2^{(k+1)nm}$$
where $y_i$ is the output of the $k+2$ rounds $n$:$kn$-UFN permutation
generator.
\end{lemma}

\begin{proof}
From the definition of a $k+2$ rounds $n$:$kn$-UFN, the reply $y_p$ of the
$p$th oracle machine query $x_p$ is described as
    \begin{align*}
    y_p
    &=  (L_{p,1}^{k+2} \| L_{p,2}^{k+2} \| \cdots \| L_{p,k}^{k+2} \| R_p^{k+2}) \\
    &=  (R_p^{k+1} \| L_{p,2}^{k+2} \| \cdots \| L_{p,k}^{k+2} \| R_p^{k+2}) \\
    &=  (R_p^{k+1} \| (L_{p,1}^{k+1} \| \cdots \| L_{p,k-1}^{k+1} \| L_{p,k}^{k+1})
        \oplus h_{k+2}(R_p^{k+1})) \\
    &=  (L_{p,k}^k \oplus C_k(h_{k+1}(R_p^k)) \|
        (L_{p,1}^{k+1} \| \cdots \| L_{p,k-1}^{k+1} \| L_{p,k}^{k+1})
        \oplus h_{k+2}(R_p^{k+1}))
    \end{align*}
where $h_{k+1}, h_{k+2}$ are functions generated by an ideal random function
generator whose inputs are $n$ bits and whose outputs are $kn$ bits.

By the definition of the BAD event $\xi$, we obtain the following equation
    \begin{align*}
    \neg \xi
    =   \wedge_{i=1}^{k+1} \neg \xi^i
    =   (R_p^k \neq R_q^k) \wedge (R_p^{k+1} \neq R_q^{k+1})
    \end{align*}
where $q$ and $p$ are the oracle indexes with $1 \leq p < q \leq m$. Thus the
output values of $h_{k+1}, h_{k+2}$ become values with uniform distribution
since the input values ​​of $h_{k+1}, h_{k+2}$ are different for all oracle
queries and $h_{k+1}, h_{k+2}$ are functions generated by the ideal random
function generator.
\end{proof}

\begin{lemma} \label{lem:nkn_ufn_bad}
The probability of the BAD event in a $k+2$ rounds $n$:$kn$-UFN permutation
generator is bounded by
    $$\Pr(\xi) \leq m^2 / 2^n.$$
\end{lemma}

\begin{proof}
For the proof, we should show that $\Pr(\xi^k) \leq m^2 / 2^{n+1}$ and
$\Pr(\xi^{k+1}) \leq m^2 / 2^{n+1}$ since $Pr(\xi) = \Pr(\xi^k) +
\Pr(\xi^{k+1})$ by the definition of the BAD event $\xi$.

First, we show that $\Pr(\xi^k) \leq m^2 / 2^{n+1}$. By the definition of
$\xi^k$, we have $\Pr(\xi^k) = \Pr(R_p^k = R_q^k)$ for the oracle query
indexes $p$ and $q$ with $1 \leq p < q \leq m$. Thus we obtain the following
equation
    \begin{align*}
    \Pr(R_p^k = R_q^k)
    &=
    \left\{ \begin{array}{ll}
        \Pr(L_{p,k}^{k-1} \| R_p^{k-1} = L_{q,k}^{k-1} \| R_q^{k-1}) &
        \text{ if } L_{p,k}^{k-1} \| R_p^{k-1} = L_{q,k}^{k-1} \| R_q^{k-1} \\
        1/2^n & \text{ otherwise. }
    \end{array} \right.
    \end{align*}
Similarly, we can obtain the following equation
    \begin{align*}
    &\Pr(L_{p,i+1}^{i} \| \cdots \| R_p^{i} = L_{q,i+1}^{i} \| \cdots \| R_q^{i}) \\
    &=
    \left\{ \begin{array}{ll}
        \Pr(L_{p,i}^{i-1} \| \cdots \| R_p^{i-1} = L_{q,i}^{i-1} \| \cdots \|
        R_q^{i-1}) & \text{ if }
            L_{p,i}^{i-1} \| \cdots \| R_p^{i-1} = L_{q,i}^{i-1} \|
            \cdots \| R_q^{i-1} \\
        1/2^{(k+1-i)n} & \text{ otherwise }
    \end{array} \right.
    \end{align*}
Since each oracle query should be different, we have $\Pr(L_{p,1}^0 \| \cdots
\| L_{p,k}^0 \| R_p^0 = L_{q,1}^0 \| \cdots \| L_{q,k}^0 \| R_q^0) = 0$ for
oracle query indexes $p$ and $q$ with $1 \leq p < q \leq m$. Thus, we have
    \begin{align*}
    \Pr(\xi^k)
    =   {}_mC_2 \cdot \Pr(R_p^k = R_q^k)
    =   {}_mC_2 \left( \frac{1}{2^n} + \cdots + \frac{1}{2^{kn}} \right)
    \leq \frac{m^2}{2^{n+1}}.
    \end{align*}
By using similar approach, we have $\Pr(\xi^{k+1}) \leq m^2 / 2^{n+1}$.
Therefore, we obtain $\Pr(\xi) = \Pr(\xi^k) + \Pr(\xi^{k+1}) \leq m^2 /
2^{n}$.
\end{proof}

In Theorem \ref{thm:k+2_nkn_ufn_prp}, we prove that a $k+2$ rounds
$n$:$kn$-UFN that uses ideal random functions is pseudo-random from the above
two lemmas, and we also prove that a $k+2$ rounds $n$:$kn$-UFN that uses
pseudo-random functions is also pseudo-random.

\begin{theorem} \label{thm:k+2_nkn_ufn_prp}
A $k+2$ rounds $n$:$kn$-UFN permutation generator using a pseudo-random
function generator is a pseudo-random permutation generator.
\end{theorem}

\begin{proof}
From Lemma \ref{lem:nkn_ufn_not_bad} and Lemma \ref{lem:nkn_ufn_bad}, we can
show that an $n$:$kn$-UFN permutation generator $P$ that uses an ideal random
function generator is a pseudo-random permutation generator. First, we have
$| \Pr (M^P(1^{(k+1)n}) = 1 | \neg \xi) - \Pr (M^K (1^{(k+1)n}) = 1) | = 0$
since $P$ is the same as an ideal random permutation generator $K$ when the
BAD event $\xi$ does not occur by Lemma \ref{lem:nkn_ufn_not_bad}. We also
have $| \Pr (M^P(1^{(k+1)n}) = 1 | \xi) - \Pr (M^K (1^{(k+1)n}) = 1) | \leq
1$ since the absolute value of the probability difference is less than 1.
Therefore, we have the following equation
    \begin{align*}
    &   \big| \Pr (M^P(1^{(k+1)n}) = 1) - \Pr (M^K(1^{(k+1)n}) = 1) \big| \\
    &=  \big| \Pr (M^P(1^{(k+1)n}) = 1 | \xi) - \Pr (M^K(1^{(k+1)n}) = 1) \big|
        \cdot \Pr(\xi) + \\
    &\quad~
        \big| \Pr (M^P(1^{(k+1)n}) = 1 | \neg \xi) - \Pr (M^K(1^{(k+1)n}) = 1) \big|
        \cdot \Pr(\neg \xi) \\
    &\leq \Pr(\xi)
     \leq m^2 / 2^{n}.
    \end{align*}

We now show that a $k+2$ rounds $n$:$kn$-UFN permutation generator $P$ using
a pseudo-random function generator is a pseudo-random permutation generator.
For this, we use the proof by contradiction. That is, if a $k+2$ rounds
$n$:$kn$-UFN using an ideal random function generator is pseudo-random but a
$k+2$ rounds $n$:$kn$-UFN using a pseudo-random function generator is not
pseudo-random, then we can derive a contradiction to the pseudo-randomness of
the pseudo-random function generator.

Suppose that an $n$:$kn$-UFN permutation generator $P$ using a pseudo-random
function generator is not pseudo-random. Then there exists an oracle machine
$M$ which distinguishes an ideal random permutation generator $K$ and $P$
with a probability greater than $1/n^c$ for a constant $c$.

First, we let $D_{f_{k+2}} \circ \cdots \circ D_{f_{i+1}} \circ D_{h_i} \circ
\cdots \circ D_{h_1}(\cdot)$ be a permutation generator in which an ideal
random function generator is used from the first round to the $i$th round and
a pseudo-random function generator is used from the $i+1$th round to the
$k+2$th round in the $n$:$kn$-UFN permutation generator for $i$ with $0 \leq
i \leq k+2$. Let $p_i^D$ be the probability that an oracle machine that has
access to this permutation generator will output $1$. That is,
    \begin{align*}
    p_i^D
    =   \Pr (M^{D_{f_{k+2}} \circ \cdots \circ D_{f_{i+1}} \circ D_{h_i}
        \circ \cdots \circ D_{h_1}} (1^{(k+1)n}) = 1)
    \end{align*}
where $f_{i+1}, \ldots, f_{k+2}$ are functions generated by a pseudo-random
function generator and $h_1, \ldots, h_i$ are functions generated by an ideal
random function generator. Let $p^K$ be the probability that an oracle
machine that has access to the ideal random permutation generator will output
$1$. Since we supposed that the $k+2$ rounds $n$:$kn$-UFN using a
pseudo-random function generator is not pseudo-random, we get the following
equation
    \begin{align*}
    1/n^c
    &\leq | p^K - p_0^D | \\
    &\leq | p^K - p_{k+2}^D | + | p_{k+2}^D - p_{k+1}^D | + \cdots +
          | p_{i+1}^D - p_{i}^D | + \cdots + | p_1^D - p_0^D |.
    \end{align*}

However, since the $k+2$ rounds $n$:$kn$-UFN using an ideal random function
generator has been shown to be a pseudo-random permutation generator, we have
$| P^K - P_{k+2}^D | \leq m^2 / 2^{n}$. Therefore, we have $| p_{i+1}^D -
p_{i}^D | \geq 1 / (k+3)n^c$ for some $i$. By using this, an oracle machine
$C$ that distinguishes an ideal random function generator and a pseudo-random
function generator with a probability higher than $1 / (k+3)n^c$ can be
constructed as follows.

In the oracle machine $C$, we first change the part of $C$ that calculates
the oracle query reply to $D_{f_{k+2}} \circ \cdots \circ D_{f_{i+1}} \circ
D_X \circ D_{h_{i-1}} \circ \cdots \circ D_{h_1} (\cdot)$. In this case, $X$
is a function whose input is $n$ bits and whose output is $kn$ bits. If $X$
in $C$ is a function generated by a pseudo-random function generator $F$,
then we have $p_{i}^D = \Pr (C^F (1^{kn}) = 1)$. On the other hand, if $X$ in
$C$ is a function generated by an ideal random function generator $H$, then
we have $p_{i+1}^D = \Pr (C^H (1^{kn}) = 1)$. However, it is possible to
distinguish the ideal random function generator and the pseudo-random
function generator with a probability greater than $1/(k+3)n^c$ since $|
p_{i+1}^D - p_{i}^D | \geq 1/(k+3)n^c$. But this contradicts that the
pseudo-random function generator is pseudo-random. Therefore, the $k+2$
rounds $n$:$kn$-UFN using the pseudo-random function generator is a
pseudo-random permutation generator.
\end{proof}

From Theorem \ref{thm:k+1_nkn_ufn_not_prp} and Theorem
\ref{thm:k+2_nkn_ufn_prp}, the minimum number of rounds of the $n$:$kn$-UFN
permutation generator using a pseudo-random function generator to be
pseudo-random is $k+2$.

\chapter{Analysis of Modified Unbalanced Feistel Networks}
\label{chap:analysis-ufn2}

In this section, we propose an unbalanced Feistel network structure that can
extend an $n$-bit block cipher to a $(k+1)n$-bit block cipher. Then we
analyze the conditions for this modified Feistel structure to be
pseudo-random. This chapter is summarized as follows.

An $n$:$kn$-UFN2 structure, which can extend an $n$-bit block cipher to a
$(k+1)n$-bit block cipher, is a target-heavy unbalanced Feistel network where
both the input and output of round functions are $n$ bits. The main advantage
of the $n$:$kn$-UFN2 structure is that it allows to build a new block cipher
with large input size by using an existing block cipher with proven security.
In order for an $n$:$kn$-UFN2 permutation generator that uses pseudo-random
functions to be pseudo-random, the round number $k$ must be odd and the total
number of rounds must be at least $2k+1$.

The structure of this chapter is as follows. In Section 4.1, we define an
$n$:$kn$-UFN2 structure and examines the properties of this structure. In
Section 4.2, we analyze the conditions for the $n$:$kn$-UFN2 structure to be
a pseudo-random permutation generator.

\section{Definition and Property} \label{sec:ufn2-def-and-prop}

As the information throughput increases, a block cipher with a large input is
needed to process large amounts of data. In general, however, it is not easy
to build a block cipher with a large input. It is also difficult to ensure
the security of this new block cipher. However, if we can use a block cipher
that is as secure as DES, we can trust the security of the new cipher.
However, the DES cipher does not expand to a block cipher with a larger input
because the input is fixed at 64 bits. Therefore, we propose an $n$:$kn$-UFN2
permutation generator that can extend an $n$-bit block cipher to a
$(k+1)n$-bit block cipher. We also analyze the condition for this structure
to be a pseudo-random permutation generator.

\begin{figure}
\centering
\includegraphics[scale=0.75]{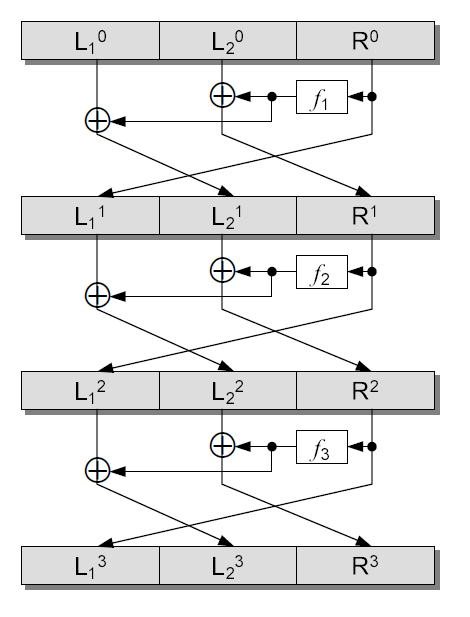}
\caption{The structure of a 3 rounds $n$:$2n$-UFN2}
\label{fig:3r_n2n_ufn2}
\end{figure}

An $n$:$kn$-UFN2 structure is a target-heavy unbalanced Feistel network in
which the input and output of round functions are $n$ bits. It is defined as
follows. For instance, a 3 rounds $n$:$2n$-UFN2 structure is described in
Figure \ref{fig:3r_n2n_ufn2}.

\begin{definition}[$n$:$kn$-UFN2] \label{def:nkn-ufn2}
For any function $f$ belonging to the set of functions $F:I_{n} \rightarrow
I_{n}$, one round of $n$:$kn$-UFN2 is defined by the following permutation
    $$D_f (L_1 \| \cdots \| L_k \| R)
    \stackrel{def}{=} (R \| L_1 \oplus f(R) \| \cdots \| L_k \oplus f(R)).$$
Similarly, for any functions $f_1, f_2, \ldots, f_r$ belonging to the set of
functions $F:I_{n} \rightarrow I_{n}$, an $r$ rounds $n$:$kn$-UFN2 is defined
by the following permutation
    $$D_{f_r} \circ \cdots \circ D_{f_2} \circ D_{f_1}
        (L_1^0 \| \cdots \| L_k^0 \| R^0)
    \stackrel{def}{=}
    D_{f_r} \circ \cdots \circ D_{f_2} ( D_{f_1}
        (L_1^0 \| \cdots \| L_k^0 \| R^0) ).$$
In this case, we have $|L_i| = |R| = n$.
\end{definition}

\section{The Pseudo-Random Proof of $n$:$kn$-UFN2} \label{sec:analysis-nkn-ufn2}

In Theorem \ref{thm:even_nkn_ufn2_not_prp}, we show that an $n$:$kn$-UFN2
permutation generator is not pseudo-random if $k$ is even. In Theorem
\ref{thm:odd_2k_nkn_ufn2_not_prp}, we show that an $n$:$kn$-UFN2 permutation
generator is not pseudo-random if $k$ is odd and the number of rounds is less
than or equal to $2k$.

\begin{theorem} \label{thm:even_nkn_ufn2_not_prp}
If $k$ is even, then an $n$:$kn$-UFN2 permutation generator is not
pseudo-random.
\end{theorem}

\begin{proof}
If $k$ is even, then we obtain the following equation
    \begin{align*}
    & L_1^{r} \oplus L_2^r \oplus \cdots \oplus L_k^r \oplus R^r \\
    &\stackrel{(a)}{=}
        R^{r-1} \oplus (L_1^{r-1} \oplus f_r(R^{r-1})) \oplus \cdots \oplus
        (L_k^{r-1} \oplus f_r(R^{r-1})) \\
    &\stackrel{(b)}{=}
        L_1^{r-1} \oplus L_2^{r-1} \oplus \cdots L_k^{r-1} \oplus R^{r-1} \db \\
    &\quad \vdots \db \\
    &\stackrel{(c)}{=}
        L_1^0 \oplus L_2^0 \oplus \cdots L_k^0 \oplus R^0.
    \end{align*}
In this equation, $(a)$ is satisfied from the definition of an $n$:$kn$-UFN2,
$(b)$ is satisfied from the property of $\oplus$ and the number of $f$ is
even, and $(c)$ is satisfied from the recursive application of the above
equation.

We can build an oracle machine $M$ that distinguishes between an
$n$:$kn$-UFN2 permutation generator and an ideal random permutation
generator. First, the oracle machine queries $x = (L_1^0 \| \cdots \| L_k^0
\| R^0)$ and receives a response $y = (L_1^r \| \cdots \| L_k^r \| R^r)$. If
the equation $L_1^0 \oplus \cdots L_k^0 \oplus R^0 = L_1^r \oplus \cdots
\oplus L_k^r \oplus R^r$ is satisfied, then the oracle machine outputs $1$.
Otherwise, it outputs $0$. Thus, if the oracle response is generated by the
$n$:$kn$-UFN2, then the output of the oracle machine is always $1$. However,
if the oracle response is generated by the ideal random permutation
generator, then the probability that the equation is satisfied is $1/2^n$.
Therefore the $n$:$kn$-UFN2 permutation generator with even $k$ is not
pseudo-random since the oracle machine $M$ exists.
\end{proof}

\begin{theorem} \label{thm:odd_2k_nkn_ufn2_not_prp}
If $k$ is odd, then a $2k$ rounds $n$:$kn$-UFN2 permutation generator is not
pseudo-random.
\end{theorem}

\begin{proof}
To prove that a $2k$ rounds $n$:$kn$-UFN2 is not pseudo-random, we show that
there is a relation between the input and output of oracle queries. Next, we
build an oracle machine that can distinguish an $n$:$kn$-UFN2 permutation
generator and an ideal permutation generator by using this relation.

From the definition of an $n$:$kn$-UFN2, we obtain the following equation
    \begin{align*}
    & L_1^{2k} \oplus \cdots \oplus L_k^{2k} \oplus R^{2k} \\
    &=  (L_1^{2k-1} \oplus \cdots \oplus L_k^{2k-1} \oplus R^{2k-1}) \oplus
        f_{2k}(R^{2k-1})) \\
    &\quad \vdots \\
    &=  (L_1^0 \oplus \cdots L_k^0 \oplus R^0) \oplus
        \bigoplus_{i=1}^{2k} f_i(R^{i-1}).
    \end{align*}
In this case, $R^{2k}$ is represented as
    \begin{align*}
    R^{2k}
    &=  R^{2k-(k+1)} \oplus (f_{k+1}(R^k) \oplus \cdots \oplus f_{2k}(R^{2k-1})) \\
    &\quad \vdots \\
    &=  W \oplus \bigoplus_{i=1}^{2k} f_i(R^{i-1}) \oplus f_k(R^{k-1})
    \end{align*}
where $W$ is defined as
    \begin{align*}
    W =
    \left\{ \begin{array}{ll}
        L_i^0   & \text{ if } 3(r+1) \mod (k+1) = i-1, \\
        R^0     & \text{ if } 3(r+1) \mod (k+1) = k.
    \end{array} \right.
    \end{align*}
Thus, we obtain the following relation between the input and output of an
$n$:$kn$-UFN2 as
    \begin{align*}
    (L_1^{2k} \oplus \cdots \oplus L_k^{2k}) \oplus
        (L_1^0 \oplus \cdots \oplus L_k^0 \oplus R^0) \oplus W
        = f_k(R^{k-1}).
    \end{align*}
If we choose two oracle queries with different $L_{p,1}^0$ and $L_{q,1}^0$
blocks for the indexes $p$ and $q$ with $1 \leq p < q \leq m$, then we obtain
the following equations
    \begin{align*}
    & (L_{p,1}^{2k} \oplus \cdots \oplus L_{p,k}^{2k}) \oplus
        (L_{p,1}^0 \oplus L_2^0 \oplus \cdots \oplus L_k^0 \oplus R^0) \oplus W_p
        = f_k(R^{k-1}), \\
    & (L_{q,1}^{2k} \oplus \cdots \oplus L_{q,k}^{2k}) \oplus
        (L_{q,1}^0 \oplus L_2^0 \oplus \cdots \oplus L_k^0 \oplus R^0) \oplus W_q
        = f_k(R^{k-1}).
    \end{align*}
By using these two equations, we have the following relation for two oracle
queries and their responses as
    \begin{align*}
    (L_{p,1}^{2k} \oplus \cdots \oplus L_{p,k}^{2k}) \oplus
        (L_{q,1}^{2k} \oplus \cdots \oplus L_{q,k}^{2k}) \oplus
        (L_{p,1}^0 \oplus L_{q,1}^0) \oplus
        (W_p \oplus W_q) = 0.
    \end{align*}

An oracle machine $M$ that distinguishes between an $n$:$kn$-UFN2 permutation
generator and an ideal random permutation generator can be built as follows.
First, the oracle machine queries $x_p$ and $x_q$ and receives responses
$y_p$ and $y_q$. If the equation $(L_{p,1}^{2k} \oplus \cdots \oplus
L_{p,k}^{2k}) \oplus (L_{q,1}^{2k} \oplus \cdots \oplus L_{q,k}^{2k}) \oplus
(L_{p,1}^0 \oplus L_{q,1}^0) \oplus (W_p \oplus W_q) = 0$ is satisfied, then
the oracle machine outputs $1$. Otherwise, it outputs $0$. Thus, if the
oracle response is generated by the $n$:$kn$-UFN2, then the output of the
oracle machine is always $1$. However, if the oracle response is generated by
the ideal random permutation generator, then the probability that the
equation is satisfied is $1/2^n$. Therefore the $n$:$kn$-UFN2 permutation
generator is not pseudo-random since the oracle machine $M$ can distinguish
two permutation generators.
\end{proof}

We now prove that a $2k+1$ rounds $n$:$kn$-UFN2 permutation generator using
pseudo-random functions is pseudo-random. First, we define an event that can
be used to distinguish a $2k+1$ rounds $n$:$kn$-UFN2 using ideal random
functions from an ideal random permutation generator as the following BAD
event.

\begin{definition}(The BAD event $\xi$ of $2k+1$ rounds $n$:$kn$-UFN2)
\label{def:bad-2k+1-nkn-ufn2} %
A random variable $\xi^i$ is defined as an event in which $R_{p}^i$ and
$R_{q}^i$ of two oracle query indexes $p$ and $q$ are equal where $1 \leq p <
q \leq m$. The BAD event $\xi$ is a random variable defined as
$\vee_{i=k}^{2k} \xi^i$.
\end{definition}

If an oracle machine is able to distinguish between a $2k+1$ rounds
$n$:$kn$-UFN2 and an ideal random permutation generator, then the BAD event
must occur. This means that if the BAD event does not occur, then the $2k+1$
rounds $n$:$kn$-UFN2 is equal to an ideal random permutation generator. In
the following lemma, we prove it.

\begin{lemma} \label{lem:nkn_ufn2_not_bad}
A $2k+1$ rounds $n$:$kn$-UFN2 permutation generator using an ideal random
function generator is equal to an ideal random permutation generator if $k$
is odd and the BAD event does not occur. That is, for all possible $\sigma_1,
\sigma_2, \ldots, \sigma_m \in \bits^{(k+1)n}$, we have
    $$\Pr (\wedge_{i=1}^m (y_i = \sigma_i) | \neg \xi) = 1/2^{(k+1)nm}$$
where $y_i$ is the output of the $2k+1$ rounds $n$:$kn$-UFN2 permutation
generator.
\end{lemma}

\begin{proof}
If the response $y_p$ of the oracle query $x_p$ is generated by an
$n$:$kn$-UFN2 permutation generator, then the oracle response is calculated
by the following algorithm. In this case, $x_p = (L_{p,1}^0 \| L_{p,2}^0 \|
\cdots \| L_{p,k}^0 \| R_p^0)$ and $y_p = (L_{p,1}^{2k+1} \| L_{p,2}^{2k+1}
\| \cdots \| L_{p,k}^{2k+1} \| R_p^{2k+1})$.
    \begin{align*}
    & \textbf{input } (L_{p,1}^0 \| L_{p,2}^0 \| \cdots \| L_{p,k}^0 \| R_p^0) \\
    & \textbf{select } h_p^1, h_p^2, \ldots, h_p^{2k+1} \in
            \text{ uniform distribution} \\
    & \textbf{for } 1 \leq i \leq 2k+1 \textbf{ do } \\
    &\quad \ell \leftarrow \text{ min } q: 1 \leq q \leq p \text{ and }
            R_q^{i-1} = R_p^{i-1} \\
    &\quad L_{p,1}^i \leftarrow R_p^{i-1} \\
    &\quad L_{p,2}^i \leftarrow L_{p,1}^{i-1} \oplus h_\ell^{i} \\
    &\quad \vdots \db \\
    &\quad L_{p,k}^i \leftarrow L_{p,k-1}^{i-1} \oplus h_\ell^{i} \db \\
    &\quad R_{p}^i \leftarrow L_{p,k}^{i-1} \oplus h_\ell^{i} \\
    & \textbf{end for} \\
    & \textbf{output } (L_{p,1}^{2k+1} \| L_{p,2}^{2k+1} \| \cdots \|
            L_{p,k}^{2k+1} \| R_p^{2k+1}).
    \end{align*}
By using this algorithm, the oracle response $y_p$ can be represented by
$h_p^k, h_p^{k+1}, \ldots, h_p^{2k+1}$ and the intermediate value
$(L_{p,1}^{k-1} \| L_{p,2}^{k-1} \| \cdots \| L_{p,k}^{k-1} \| R_p^{k-1})$
that is used in the algorithm as follows.
    \begin{align*}
    \left( \begin{array}{c}
    L_{p,1}^{k-1} \\ L_{p,2}^{k-1} \\ L_{p,3}^{k-1} \\ \cdots \\
    L_{p,k-1}^{k-1} \\ L_{p,k}^{k-1} \\ R_{p}^{k-1} \\
    \end{array} \right)
    \oplus
    \left( \begin{array}{ccccccc}
    1 & 1 & 1 & \cdots & 1 & 1 & 0 \\
    1 & 1 & 1 & \cdots & 1 & 0 & 1 \\
    1 & 1 & 1 & \cdots & 0 & 1 & 1 \\
    \multicolumn{7}{c}{\dotfill} \\
    1 & 1 & 0 & \cdots & 1 & 1 & 1 \\
    1 & 0 & 1 & \cdots & 1 & 1 & 1 \\
    0 & 1 & 1 & \cdots & 1 & 1 & 1 \\
    \end{array} \right)
    \left( \begin{array}{c}
    h_p^k \\ h_p^{k+1} \\ h_p^{k+2} \\ \cdots \\
    h_p^{2k-1} \\ h_p^{2k} \\ h_p^{2k+1} \\
    \end{array} \right)
    =
    \left( \begin{array}{c}
    L_{p,1}^{2k+1} \\ L_{p,2}^{2k+1} \\ L_{p,3}^{2k+1} \\ \cdots \\
    L_{p,k-1}^{2k+1} \\ L_{p,k}^{2k+1} \\ R_{p}^{2k+1} \\
    \end{array} \right).
    \end{align*}
Let $y^{k-1} = (L_{p,1}^{k-1}, \cdots, L_{p,k}^{k-1}, R_p^{k-1})$, $y^{2k+1}
= (L_{p,1}^{2k+1}, \cdots, L_{p,k}^{2k+1}, R_p^{2k+1})$, and $x = (h_p^k,
h_p^{k+1}, \cdots, h_p^{2k+1})$ be column vectors. Let $A$ be a $(k+1) \times
(k+1)$ matrix in which only the elements $A_{i,k+2-i}$ are $0$ and the other
elements are $1$. Then the above equation can be represented as a linear
relation $y^{k-1} \oplus A x = y^{2k+1}$. If we let $y = y^{k-1} \oplus
y^{2k+1}$, then the linear relation can be represented as $Ax = y$. This is
equal to a function $T: I_{(k+1)n} \rightarrow I_{(k+1)n}$ such that $T(x) =
Ax$.

Now, we show that the output $y$ of $T(x) = Ax$ has a uniform probability
distribution. The value $x$ ​​has a uniform probability distribution since the
$n$:$kn$-UFN2 permutation generator uses ideal random functions. That is, the
probability that $x_i$ is chosen in $I_{(k+1)n}$ is $1/2^{(k+1)n}$.
Therefore, if the function $T$ is a one-to-one correspondence function, then
the output $y_i$ value of the function $T$ also has a probability value
$1/2^{(k+1)n}$ equal to $x_i$. If the function $T$ is a one-to-one function,
then $T$ is a one-to-one correspondence function since the domain and region
of $T$ are the same.

If the function $T(x) = y$ is a one-to-one correspondence, then the inverse
of $T$ must exist. That is, the inverse matrix $A^{-1}$ of the matrix $A$
must exist since $T(x) = Ax$. If the determinant of $A$ is not zero, then
there exists an inverse matrix. So we should show that $det (A) \neq 0$.

In order to show that the determinant of the square matrix $A$ is not zero,
we convert $A$ to an echelon form $A'$. At this time, only row addition and
row interchange are used. If the echelon form $A'$ contains a row of zero,
then $det(A) = 0$, otherwise $det(A) \neq 0$ \cite{FraleighB95}. Thus, we
should show that $A'$ does not contain a row of zero.

First, we select an element in $(1,1)$ as the first pivot element and perform
row addition. Then we have the following matrix. We also select an element in
$(i, k+2-i)$ as $i$th pivot element for $2 \leq i \leq k$.
    \begin{align*}
    \left( \begin{array}{ccccccc}
    1 & 1 & 1 & \cdots & 1 & 1 & 0 \\
      &   &   &        &   & \underline{1} & 1 \\
      & O &   &        & \underline{1} &   & 1 \\
      &   &   & \iddots &   &   &   \\
      &   & \underline{1} &        & O &   & 1 \\
      & \underline{1} &   &        &   &   & 1 \\
    0 & 1 & 1 & \cdots & 1 & 1 & 1 \\
    \end{array} \right).
    \end{align*}
The above matrix contains pivot elements for all columns except the $k+1$
column. Therefore, the element at position $(k+1, k+1)$ should be selected as
the $k+1$th pivot. To do so, all values ​​except the $k+1$th element of the
$k+1$th row vector must be set to zero. Let $r_i$ be the $i$th row vector of
the above matrix. The $k+1$th pivotal element is obtained by performing row
addition $r_2 \oplus r_3 \oplus \cdots \oplus r_k \oplus r_{k+1}$. By the
property  of the XOR operation, the result of row addition is $(0, 0, \ldots,
0, 0, X)$, where $X = 0$ if $k$ is even and $X = 1$ if $k$ is odd depending
on the property of the XOR. An echelon form $A'$ is obtained by performing
row exchange so that each pivot is located diagonally in the matrix.

Thus, when $k$ is an odd number, the echelon form matrix $A'$ does not
include a row of zero. Therefore, if $k$ is odd, the output of $2k+1$ rounds
$n$:$kn$-UFN2 has a uniform probability distribution.
\end{proof}

\begin{lemma} \label{lem:nkn_ufn2_bad}
The probability of the BAD event in a $2k+1$ rounds $n$:$kn$-UFN2 permutation
generator is bounded by
    $$\Pr(\xi) \leq (k+1)m^2 / 2^{n+1}.$$
\end{lemma}

\begin{proof}
For the proof, we should show that $\Pr(\xi^i) \leq m^2 / 2^{n+1}$ since
$\Pr(\xi) = \vee_{i=k}^{2k} \Pr(\xi^i)$ by the definition of the BAD event
$\xi$.

By the definition of $\xi^i$, we have $\Pr(\xi^i) = \Pr(R_p^i = R_q^i)$ for
the oracle query indexes $p$ and $q$ with $1 \leq p < q \leq m$. Thus we
obtain the following equation
    \begin{align*}
    \Pr(R_p^i = R_q^i)
    &=
    \left\{ \begin{array}{ll}
        \Pr(L_{p,k}^{i-1} \| R_p^{i-1} = L_{q,k}^{i-1} \| R_q^{i-1}) &
        \text{ if } L_{p,k}^{i-1} \| R_p^{i-1} = L_{q,k}^{i-1} \| R_q^{i-1} \\
        1/2^n & \text{ otherwise. }
    \end{array} \right.
    \end{align*}
Similarly, we can obtain the following equation
    \begin{align*}
    & \Pr(L_{p,k-j+1}^{i-j} \| \cdots \| R_p^{i-j} =
        L_{q,k-j+1}^{i-j} \| \cdots \| R_q^{i-j}) \\
    &=
    \left\{ \begin{array}{ll}
        \Pr(L_{p,k-j}^{i-j-1} \| \cdots \| R_p^{i-j-1} =
            L_{q,i}^{k-j} \| \cdots \| R_q^{i-j}) &
            \text{ if }
            L_{p,k-j}^{i-j-1} \| \cdots \| R_p^{i-j-1} =
                L_{q,k-j}^{i-j-1} \| \cdots \| R_q^{i-j-1} \\
        1/2^{(j+i)n} & \text{ otherwise }
    \end{array} \right.
    \end{align*}
Since each oracle query should be different, we have $\Pr(L_{p,1}^0 \| \cdots
\| L_{p,k}^0 \| R_p^0 = L_{q,1}^0 \| \cdots \| L_{q,k}^0 \| R_q^0) = 0$ for
oracle query indexes $p$ and $q$ with $1 \leq p < q \leq m$. Thus, we have
    \begin{align*}
    \Pr(\xi^i)
    &=   {}_mC_2 \cdot \Pr(R_p^i = R_q^i) \\
    &=   {}_mC_2 \left( \frac{1}{2^n} + \cdots + \frac{1}{2^{kn}} +
        (i-k) \frac{1}{2^{(k+1)n}} \right)
    \leq \frac{m^2}{2^{n+1}}.
    \end{align*}
Therefore, we have $\Pr(\xi) = \sum_{i=k}^{2k} m^2 / 2^{n+1} = (k+1) m^2 /
2^{n+1}$.
\end{proof}

In Theorem \ref{thm:2k+1_nkn_ufn2_prp}, we prove that a $2k+1$ rounds
$n$:$kn$-UFN2 that uses ideal random functions is pseudo-random from the
above two lemmas, and we also prove that a $2k+1$ rounds $n$:$kn$-UFN2 that
uses pseudo-random functions is also pseudo-random.

\begin{theorem} \label{thm:2k+1_nkn_ufn2_prp}
A $2k+1$ rounds $n$:$kn$-UFN2 permutation generator using a pseudo-random
function generator is a pseudo-random permutation generator.
\end{theorem}

\begin{proof}
From Lemma \ref{lem:nkn_ufn2_not_bad} and Lemma \ref{lem:nkn_ufn2_bad}, we
can show that an $n$:$kn$-UFN2 permutation generator $P$ that uses an ideal
random function generator is a pseudo-random permutation generator. First, we
have $| \Pr (M^P(1^{(k+1)n}) = 1 | \neg \xi) - \Pr (M^K (1^{(k+1)n}) = 1) | =
0$ since $P$ is the same as an ideal random permutation generator $K$ when
the BAD event $\xi$ does not occur by Lemma \ref{lem:nkn_ufn2_not_bad}. We
also have $| \Pr (M^P(1^{(k+1)n}) = 1 | \xi) - \Pr (M^K (1^{(k+1)n}) = 1) |
\leq 1$ since the absolute value of the probability difference is less than
1. Therefore, we have the following equation
    \begin{align*}
    &   \big| \Pr (M^P(1^{(k+1)n}) = 1) - \Pr (M^K(1^{(k+1)n}) = 1) \big| \\
    &=  \big| \Pr (M^P(1^{(k+1)n}) = 1 | \xi) - \Pr (M^K(1^{(k+1)n}) = 1) \big|
        \cdot \Pr(\xi) + \\
    &\quad~
        \big| \Pr (M^P(1^{(k+1)n}) = 1 | \neg \xi) - \Pr (M^K(1^{(k+1)n}) = 1) \big|
        \cdot \Pr(\neg \xi) \\
    &\leq \Pr(\xi)
     \leq (k+1)m^2 / 2^{n+1}.
    \end{align*}

We now show that a $2k+1$ rounds $n$:$kn$-UFN2 permutation generator $P$
using a pseudo-random function generator is a pseudo-random permutation
generator. For this, we use the proof by contradiction. That is, if a $2k+1$
rounds $n$:$kn$-UFN2 using an ideal random function generator is
pseudo-random but a $2k+1$ rounds $n$:$kn$-UFN2 using a pseudo-random
function generator is not pseudo-random, then we can derive a contradiction
to the pseudo-randomness of the pseudo-random function generator.

Suppose that an $n$:$kn$-UFN2 permutation generator $P$ using a pseudo-random
function generator is not pseudo-random. Then there exists an oracle machine
$M$ which distinguishes an ideal random permutation generator $K$ and $P$
with a probability greater than $1/n^c$ for a constant $c$.

First, we let $D_{f_{2k+1}} \circ \cdots \circ D_{f_{i+1}} \circ D_{h_i}
\circ \cdots \circ D_{h_1}(\cdot)$ be a permutation generator in which an
ideal random function generator is used from the first round to the $i$th
round and a pseudo-random function generator is used from the $i+1$th round
to the $2k+1$th round in an $n$:$kn$-UFN2 permutation generator for $i$ with
$0 \leq i \leq 2k+1$. Let $p_i^D$ be the probability that an oracle machine
that has access to this permutation generator will output $1$. That is,
    \begin{align*}
    p_i^D
    =   \Pr (M^{D_{f_{2k+1}} \circ \cdots \circ D_{f_{i+1}} \circ D_{h_i}
        \circ \cdots \circ D_{h_1}} (1^{(k+1)n}) = 1)
    \end{align*}
where $f_{i+1}, \ldots, f_{2k+1}$ are functions generated by a pseudo-random
function generator and $h_1, \ldots, h_i$ are functions generated by an ideal
random function generator. Let $p^K$ be the probability that an oracle
machine that has access to the ideal random permutation generator will output
$1$. Since we supposed that the $2k+1$ rounds $n$:$kn$-UFN2 using a
pseudo-random function generator is not pseudo-random, we get the following
equation
    \begin{align*}
    1/n^c
    &\leq | p^K - p_0^D | \\
    &\leq | p^K - p_{2k+1}^D | + | p_{2k+1}^D - p_{2k}^D | + \cdots +
          | p_{i+1}^D - p_{i}^D | + \cdots + | p_1^D - p_0^D |.
    \end{align*}

However, since the $2k+1$ rounds $n$:$kn$-UFN2 using an ideal random function
generator has been shown to be a pseudo-random permutation generator, we have
$| P^K - P_{2k+1}^D | \leq (k+1)m^2 / 2^{n+1}$. Therefore, we have $|
p_{i+1}^D - p_{i}^D | \geq 1 / (2k+2)n^c$ for some $i$. By using this, an
oracle machine $C$ that distinguishes an ideal random function generator and
a pseudo-random function generator with a probability higher than $1 /
(2k+2)n^c$ can be constructed as follows.

In the oracle machine $C$, we first change the part of $C$ that calculates
the oracle query reply to $D_{f_{2k+1}} \circ \cdots \circ D_{f_{i+1}} \circ
D_X \circ D_{h_{i-1}} \circ \cdots \circ D_{h_1} (\cdot)$. In this case, $X$
is a function whose input is $n$ bits and whose output is $n$ bits. If $X$ in
$C$ is a function generated by a pseudo-random function generator $F$, then
we have $p_{i}^D = \Pr (C^F (1^{kn}) = 1)$. On the other hand, if $X$ in $C$
is a function generated by an ideal random function generator $H$, then we
have $p_{i+1}^D = \Pr (C^H (1^{kn}) = 1)$. However, it is possible to
distinguish the ideal random function generator and the pseudo-random
function generator with a probability greater than $1/(2k+2)n^c$ since $|
p_{i+1}^D - p_{i}^D | \geq 1/(2k+2)n^c$. But this contradicts that the
pseudo-random function generator is pseudo-random. Therefore, the $2k+1$
rounds $n$:$kn$-UFN2 using a pseudo-random function generator is a
pseudo-random permutation generator.
\end{proof}

\chapter{Comparison of Unbalanced Feistel Networks} \label{chap:comp-ufn}

In this chapter, we compare the amount of memory required for implementation
and the time of a pseudo-random permutation generator when implementing the
pseudo-random permutation generator using a balanced Feistel network,
unbalanced Feistel networks, and the newly proposed Feistel network.

To implement a pseudo-random permutation generator, we must implement a
pseudo-random function used in a Feistel network. However, the implementation
cost of the pseudo-random function and the speed of the pseudo-random
function greatly affect the cost and speed of the Feistel network. Therefore,
depending on how to implement the pseudo-random function, the comparison
value can vary greatly. Therefore, in this chapter, we will compare the
method that can most simply implement the pseudo-random function and the
method that uses the  GGM pseudo-random function generator. In other words,
the cost and the execution speed of each pseudo-random permutation generator
are investigated by using these two methods.

\section{Using Memory and Pseudo-Random Bit Generator}

The simplest way to implement a pseudo-random function is to use memory and a
pseudo-random bit generator. That is, if the input value $x$ of a function is
given, the output value $y$ is calculated by using a pseudo-random bit
generator. At this time, the initial value of the pseudo-random bit generator
uses an ideal random value such as coin tosses. However, a pseudo-random
function must give the same output value for the same input value. Thus, the
input value and the output value are stored in memory. That is, given the
same input value as before, the output is obtained by referring to the memory
instead of using the pseudo-random bit generator.

Before describing this algorithm, we assume that the size of the key that
specifies a pseudo-random function used in the $i$th round of a Feistel
network is $\ell$ bits. We also assume that the input of a pseudo-random bit
generator $G$ is $\ell$ bits and the output is $P_2$ bit. That is, $G:
I_{\ell} \rightarrow I_{P_2}$. An algorithm that implements a pseudo-random
function $F_{k_i}: I_{P_1} \rightarrow I_{P_2}$ whose input is $P_1$ bits and
output is $P_2$ bits is as follows.
\begin{algorithm}
\caption{Calculate $y = F_{k_i}(x)$} \label{alg:prf-from-mem-prbg}
\begin{algorithmic}[1]
    \Require{$|x| = P_1$ and $|y| = P_2$}
    \Ensure{$y = F_{k_i}(x)$}
    \Procedure{}{}
        \If {$mem[i][x] = \text{null}$}
            \State $\text{generate} \ell-\text{bit string } s \text{ using
                   coin-toss}$.
            \State $mem[i][x] \gets G(s)$.
        \EndIf
        \State $y \gets mem[i][x]$.
    \EndProcedure
\end{algorithmic}
\end{algorithm}

Since the input of the pseudo-random function is $P_1$ bits and the output is
$P_2$ bits, the amount of memory required to implement a pseudo-random
function is $2^{P_1} P_2$ bits. Therefore, an $r$ rounds Feistel network
using pseudo-random functions requires $r \cdot 2^{P_1} P_2$ bits memory. The
time to calculate the output value of a pseudo-random function is
proportional to the time to calculate the output of a pseudo-random bit
generator. Thus, the execution time of a pseudo-random function whose output
is $P_2$ bits is proportional to the $P_2$ value of the output size.
Therefore, the execution time of the $r$ rounds Feistel network is
proportional to $r \cdot P_2$ value.

\begin{table*}
\caption{Comparison of PRPs using memory and pseudo-random bit generators}
\label{tab:comp-prp-mem-prbg}
\vs \small \addtolength{\tabcolsep}{12.0pt}
\renewcommand{\arraystretch}{1.4}
\newcommand{\otoprule}{\midrule[0.09em]}
    \begin{tabularx}{6.50in}{lcccc}
    \toprule
        & Balanced & \multicolumn{2}{c}{Unbalanced Feistel network} & Proposed structure \\
        & Feistel network & $kn$:$n$-UFN & $n$:$kn$-UFN & $n$:$kn$-UFN2 \\
    \midrule
    Number of rounds & 3                & $k+2$         & $k+2$       & $2k+1$  \\
    \midrule
    Memory size      & $2^{(k+1)n/2}kn$ & $2^{kn} kn$   & $2^n k^2 n$ & $2^n kn$ \\
    Relative ratio   & $2^{(k-1)n/2}$   & $2^{(k-1)n}$  & $k$         & $1$ \\
    \midrule
    Running time     & $1.5 kn$         & $kn$          & $k^2 n$     & $2kn$ \\
    Relative ratio   & $1.5$            & $1$           & $k$         & $2$ \\
    \bottomrule
    \end{tabularx}
\end{table*}

In case of implementing a pseudo-random function by using the algorithm 1,
the memory size required for each pseudo-random permutation generator and the
execution time of each Feistel network are shown in Table
\ref{tab:comp-prp-mem-prbg}. In the total memory size, the pseudo-random
permutation generator using a newly proposed structure ($n$:$kn$-UFN2)
requires the least memory and the pseudo-random permutation generator using a
source-heavy unbalanced Feistel network ($kn$:$n$-UFN) requires the most
memory. In the execution time of the whole Feistel network, the pseudo-random
permutation generator using a source-heavy unbalanced Feistel network
($kn$:$n$-UFN) is the fastest and the pseudo-random permutation generator
using a target-heavy unbalanced Feistel network ($n$:$kn$-UFN) is the
slowest.

\section{Using the GGM Pseudo-Random Function}

Another way to implement a pseudo-random function is to use the pseudo-random
function generator of Goldreich, Goldwasser, and Micali. A major advantage of
the GGM pseudo-random function generator is that it can implement a
pseudo-random function without storing previously generated pseudo-random
bits.

Let $\ell$ be the size of the entire key of a pseudo-random permutation
generator. Then the $i$th pseudo-random function of the pseudo-random
permutation generator from an $r$ rounds Feistel network will use $\ell / r$
bits of the key. We first prepare two pseudo-random bit generators $G:
I_{\ell/r} \rightarrow I_{2\ell/r}$ and $G': I_{\ell/r} \rightarrow I_{P_2}$.
Let $G_0$ be the first $\ell/r$ bits of $G$, and $G_1$ be the next $\ell/r$
bits of $G$. A pseudo-random function $F: I_{P_1} \rightarrow I_{P_2}$ whose
the key size is $\ell / r$ is defined by the following algorithm.

\begin{algorithm}
\caption{Calculate $y = F_{k_i}(x)$} \label{alg:prf-from-ggm}
\begin{algorithmic}[1]
    \Require{$|x| = P_1$, $|k_i| = \ell/r$, and $|y| = P_2$}
    \Ensure{$y = F_{k_i}(x)$}
    \Procedure{}{}
        \For {$1 \leq j \leq P_1$}
            \If {$x[j] = 0$}
                \State $T \gets G_0(k_i)$.
            \Else
                \State $T \gets G_1(k_i)$.
            \EndIf
        \EndFor
        \State $y \gets G'(T)$.
    \EndProcedure
\end{algorithmic}
\end{algorithm}

Since the pseudo-random bits that are previously generated by the algorithm
can not be stored, the amount of memory required is fixed. By analyzing the
above algorithm, the number of pseudo-random bits that the algorithm must
generate to compute one pseudo-random permutation is maximum $(2 P_1 \ell / r
+ P_2) r$ bits. Therefore, the memory size and the time of each Feistel
network are shown in Table \ref{tab:comp-prp-ggm-prf}.

\begin{table*}
\caption{Comparison of PRPs using the GGM pseudo-random function generator}
\label{tab:comp-prp-ggm-prf}
\vs \small \addtolength{\tabcolsep}{12.0pt}
\renewcommand{\arraystretch}{1.4}
\newcommand{\otoprule}{\midrule[0.09em]}
    \begin{tabularx}{6.50in}{lcccc}
    \toprule
        & Balanced & \multicolumn{2}{c}{Unbalanced Feistel network} & Proposed structure \\
        & Feistel network & $kn$:$n$-UFN & $n$:$kn$-UFN & $n$:$kn$-UFN2 \\
    \midrule
    Number of rounds & 3        & $k+2$     & $k+2$     & $2k+1$  \\
    \midrule
    Memory size      & -        & -         & -         & - \\
    Relative ratio   & -        & -         & -         & - \\
    \midrule
    Running time     & $kn\ell$ & $2kn\ell$ & $2n\ell$  & $2n\ell$ \\
    Relative ratio   & $k$      & $2k$      & $2$       & $2$ \\
    \bottomrule
    \end{tabularx}
\end{table*}

It can be seen that the pseudo-random permutation generator using a
target-heavy ($n$:$kn$-UFN) and the proposed structure ($n$:$kn$-UFN2) is the
fastest in the execution time of the entire Feistel network. On the other
hand, it can be seen that the pseudo-random permutation generator using a
source-heavy ($kn$:$n$-UFN) is the slowest.

\chapter{Conclusion}

In this paper, we analyze the minimum number of rounds for a block cipher
from unbalanced Feistel networks to be a pseudo-random permutation generator
which is a safe and efficient block cipher. We also propose a new unbalanced
Feistel network structure that can be extended and analyze the minimum number
of rounds for this structure to be a pseudo-random permutation generator.

The minimum number of rounds for a permutation generator from unbalanced
Feistel networks to be a pseudo-random permutation generator is as follows.
\begin{itemize}
\item In case of a source-heavy unbalanced Feistel network where the source
    block is $kn$ bits and the target block is $n$ bits: If pseudo-random
    functions are used in round functions and the total number of rounds is
    $k+2$ or more, then a source-heavy unbalanced Feistel network is a
    pseudo-random permutation generator.

\item In case of a target-heavy unbalanced Feistel network where the
    source-block is $n$ bits and the target block is $kn$ bits: If
    pseudo-random functions are used in round functions and the total
    number of rounds is more than $k+2$, then a target-heavy unbalanced
    Feistel network is a pseudo-random permutation generator.
\end{itemize}
A newly proposed architecture is an unbalanced Feistel network architecture
that can extend an $n$-bit block cipher to a $(k+1)n$-bit block cipher. The
minimum number of rounds of a permutation generator from this structure to be
pseudo-random is as follows.
\begin{itemize}
\item If $k$ is even, then a newly proposed structure is not a
    pseudo-random permutation generator.

\item If $k$ is odd, pseudo-random functions are used, and the total number
    of rounds is $2k+1$ or more, then a newly proposed structure is a
    pseudo-random permutation generator.
\end{itemize}
In all three structures, the probability that an arbitrary algorithm can
distinguish between an ideal random permutation generator and a pseudo-random
permutation generator is less than $m^2 / 2^n$ after obtaining $m$ plaintext
and ciphertext pairs.
Therefore, when a block cipher is implemented using an unbalanced Feistel
network structure, it can be a safe and efficient block cipher only if the
number of rounds shown in this paper is satisfied.

In this paper, we examined only the condition of the number of rounds for an
unbalanced Feistel network using pseudo-random functions to be pseudo-random.
If a permutation generator from an unbalanced Feistel network is a super
pseudo-random permutation, it becomes a secure block cipher for both chosen
plaintext and chosen ciphertext attacks. Therefore, in the future, it is
necessary to analyze the condition of the round number for a block cipher
from an unbalanced Feistel network to be the super pseudo-random permutation
generator.

\bibliographystyle{plain}
\bibliography{prp-from-ufn}

\end{document}